\documentclass[%
 reprint,
 amsmath,amssymb,
 aps,
]{revtex4-1}
\usepackage{braket}
\usepackage[usenames, dvipsnames]{color}
\usepackage{graphicx}
\usepackage{dcolumn}
\usepackage{bm}
\usepackage{amsthm}
\usepackage{hyperref}

\newtheorem{mydef}{Definition}
\newtheorem{mythm}{Theorem}
\newtheorem{mylemm}{Lemma}
\newtheorem{mycorol}{Corollary}

\begin{document}

\preprint{APS/123-QED}

\title{Matrix Product Representation of Locality Preserving Unitaries}
\author{M. Burak \surname{\c{S}ahino\u{g}lu}}
\affiliation{Department of Physics and Institute for Quantum Information and Matter, California Institute of Technology, Pasadena, California 91125, USA}
\author{Sujeet K. \surname{Shukla}}
\affiliation{Department of Physics and Institute for Quantum Information and Matter, California Institute of Technology, Pasadena, California 91125, USA}
\author{Feng \surname{Bi}}
\affiliation{Department of Physics and Institute for Quantum Information and Matter, California Institute of Technology, Pasadena, California 91125, USA}
\author{Xie \surname{Chen}}
\affiliation{Department of Physics and Institute for Quantum Information and Matter, California Institute of Technology, Pasadena, California 91125, USA}

\date{\today}

\begin{abstract}
The matrix product representation provides a useful formalism to study not only entangled states, but also entangled operators in one dimension. In this paper, we focus on unitary transformations and show that matrix product operators that are unitary provides a necessary and sufficient representation of 1D unitaries that preserve locality. That is, we show that matrix product operators that are unitary are guaranteed to preserve locality by mapping local operators to local operators while at the same time all locality preserving unitaries can be represented in a matrix product way. Moreover, we show that the matrix product representation gives a straight-forward way to extract the GNVW index defined in Ref.\cite{Gross2012} for classifying 1D locality preserving unitaries. The key to our discussion is a set of `fixed point' conditions which characterize the form of the matrix product unitary operators after blocking sites. Finally, we show that if the unitary condition is relaxed and only required for certain system sizes, the matrix product operator formalism allows more possibilities than locality preserving unitaries. In particular, we give an example of a simple matrix product operator which is unitary only for odd system sizes, does not preserve locality and carries a `fractional' index as compared to their locality preserving counterparts.
\end{abstract}

\pacs{Valid PACS appear here}

\maketitle


\section{\label{sec:intro} Introduction}

The matrix product formalism~\cite{Fannes1992,Perez-Garcia2007} has played a significant role in the study of one dimensional systems. In particular, the matrix product representation of 1D quantum states underlies successful numerical algorithms like the Density Matrix Renormalization Group algorithm~\cite{White1993} and the Time-Evolving Block Decimation algorithm~\cite{Vidal2003}. Moreover, the matrix product representation provides a deep insight into the structure of the ground states in 1D~\cite{Perez-Garcia2007} which enables rigorous proofs of the efficiency of 1D variational algorithms in search for the ground states~\cite{Schuch2010,Landau2013arxiv} and also a complete classification of 1D gapped phases~\cite{Pollmann2010,Pollmann2012,Chen2011,Schuch2011}.

Operators can also be represented in a matrix product form~\cite{Verstraete2004,Pirvu2010,Haegeman2016}, which provides a useful tool in the simulation of one dimensional mixed states and real / imaginary time evolutions (see for example Ref.~\onlinecite{Mascarenhas2015,Wall2012}). In particular, matrix product operators which are unitary play an important role in not only the simulation of dynamical processes in 1D , but also the understanding and classification of (symmetry protected) topological phases in 2D~\cite{Chen2011a,Buerschaper14,Sahinoglu2014,Williamson2016,Bultinck2017}. 

How well does the matrix product formalism represent unitary operators in 1D? Of particular interest are unitaries that preserve the locality structure of the system. That is, unitaries that map local operators to local operators. We want to understand: Can all locality preserving 1D unitaries be represented using the matrix product form? On the other hand, of course not all matrix product opereators are unitary. But among those that are, what conditions do they have to satisfy to preserve locality? Moreover, it has been shown~\cite{Gross2012} that locality preserving 1D unitaries can be classified according to how much information they are transmitting across any cut in the 1D chain and each class can be uniquely characterized by a positive rational index, which we refer to below as the GNVW index. We want to know if there is a simple way to extract this GNVW index from the matrix product representation if such a representation exists. 

In this paper, we address the above questions and show that 
\begin{itemize}
\item Unitary matrix product operators provide a necessary and sufficient representation of locality preserving unitaries in 1D. 
\end{itemize}
That is, matrix product operators that are unitary are guaranteed to preserve locality by mapping local operators to local operators while at the same time all locality preserving unitaries can be represented in a matrix product way. Moreover, we find that
\begin{itemize}
\item The GNVW index can be extracted in a simple way as the square root of $I_{\text{RR}}$, the `Rank-Ratio index', which is the ratio between the rank of the left and right singular value decompositions of the tensor representing the operator
\begin{equation}
\begin{array}{l}
I_{\text{RR}} = \text{rank}\left(\raisebox{-.4\height}{\includegraphics[height=0.7cm]{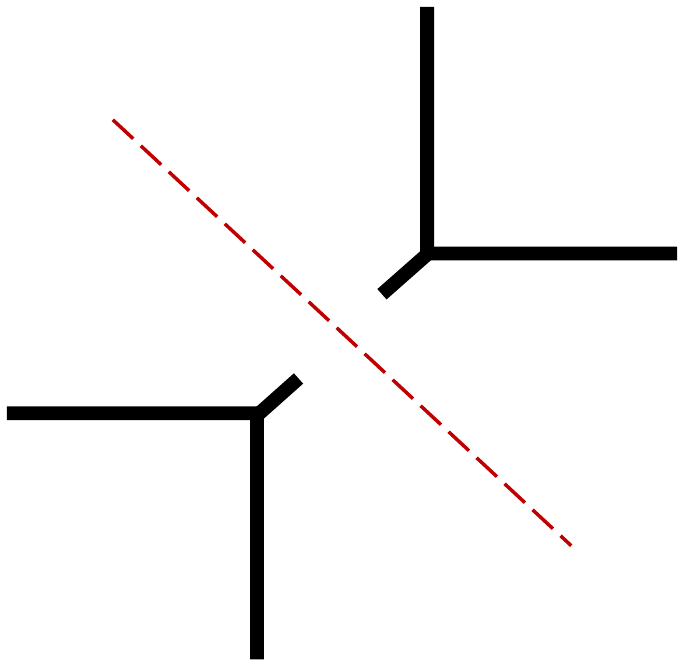}}\right) \bigg/ \text{rank}\left(\raisebox{-.4\height}{\includegraphics[height=0.7cm]{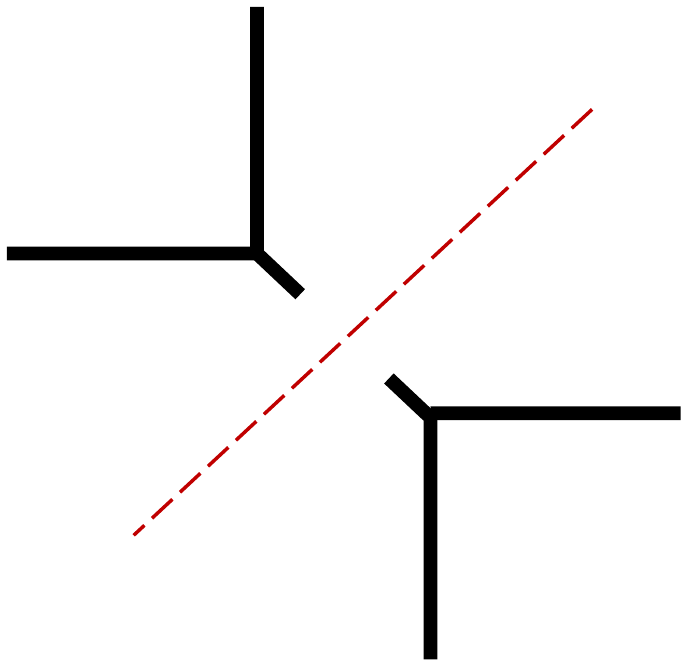}}\right)\\
I_{\text{GNVW}} = \sqrt{I_{\text{RR}}}
\label{index_intro}
\end{array}
\end{equation}
\end{itemize}
The exact meaning and a more rigorous version of this formula is given in section~\ref{sec:index}.

To show this result, we start from the basic requirements for a matrix product operator to be unitary in section~\ref{sec:MPU}. Based on these basic requirements, we prove in the section~\ref{sec:conditions} that after sufficient blocking, the `fixed point' matrix product operator satisfies a set of nice fixed point properties. Using this set of fixed point conditions, we can show the correspondence between matrix product unitary operators and locality preserving 1D unitaries. Moreover, these conditions enable us to prove in section~\ref{sec:index} that Eq.~\ref{index_intro} provides a well-defined index for each equivalence class of 1D locality preserving unitaries and it exactly matches (the square of) the GNVW index. In section~\ref{sec:numerics}, we compute the index according to Eq.~\ref{index_intro} numerically for some random locality preserving unitaries and demonstrate how it approaches the expected value as we take larger and larger blocks of the tensor. In section~\ref{sec:beyondGNVW}, we show that the matrix product formalism also provides interesting ways to go beyond the GNVW framework. In particular, we give an example of a simple matrix product operator with `fractional' index as compared to the locality preserving ones. This example does not contradict with our discussions in the previous sections because it is unitary only in systems of special sizes and does not preserve locality.

The structure of the paper is illustrated in Fig.~\ref{struc}.

\begin{figure}[htbp]
\includegraphics[scale=.6]{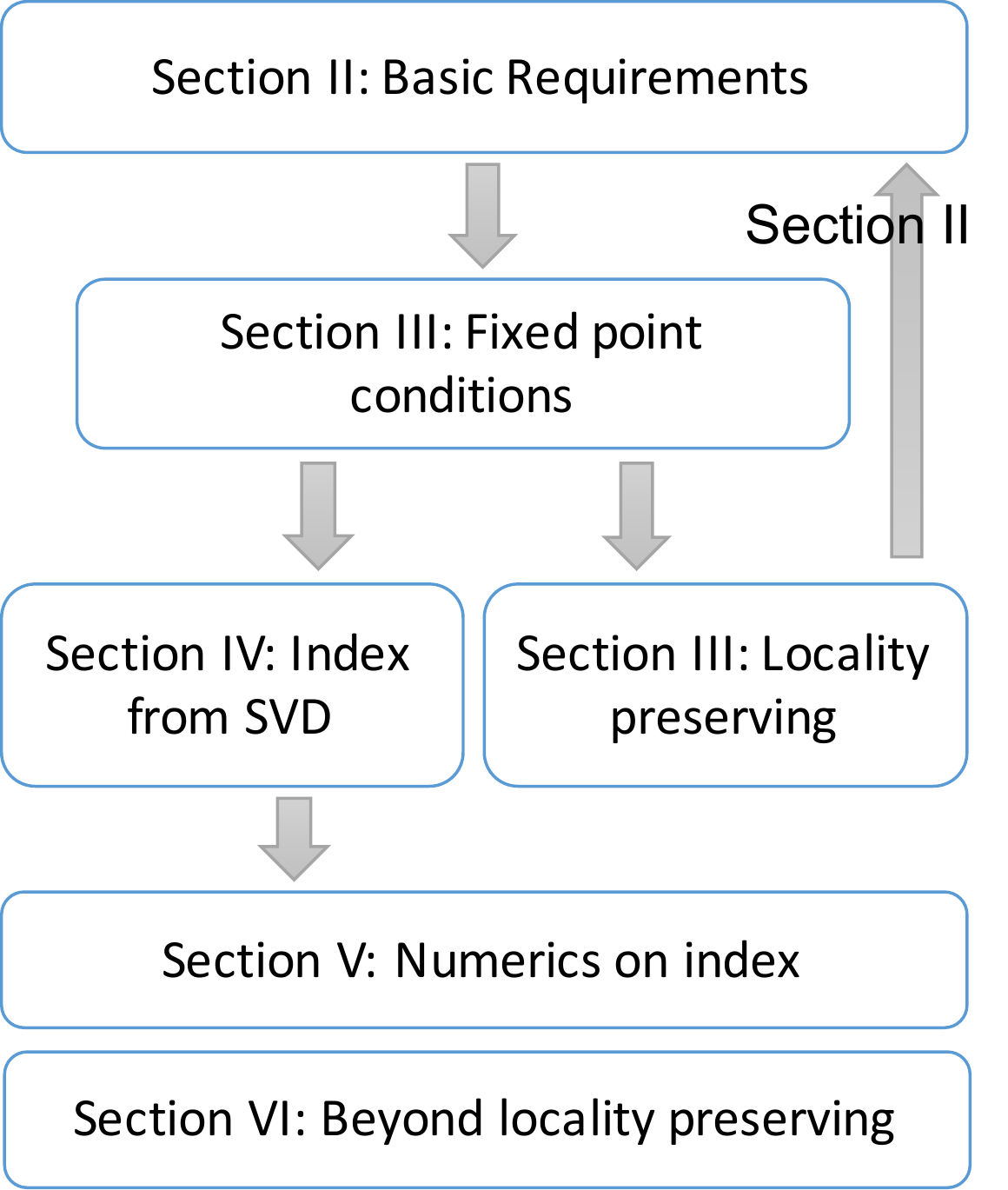}
\caption{Structure and logic of this paper.}
\label{struc}
\end{figure}

\section{\label{sec:MPU}Matrix Product Unitary: basic requirements}

\subsection{Implication of the unitary condition}

Let's first set up the stage and discuss the basic requirement a matrix product operator (MPO) has to satisfy to represent a unitary operator. Consider an MPO $O$ acting on $N$ sites where each site has a $d$-dimensional degree of freedom, i.e., $O$ acts on $(\mathbb C^d)^{\otimes N}$. In principle, $N$ is very large, ideally goes to infinity. In this paper we focus on translation invariant MPO with periodic boundary condition. The matrix product form of $O$ is given by
\begin{equation}\label{MPO}
O^{j_1j_2...j_N}_{i_1i_2...i_N} =  \text{Tr}\left(M^{j_1i_1}M^{j_2i_2}... M^{j_Ni_N}\right)
\end{equation}
where each $M^{j_ki_k}$, with fixed $i_k$ and $j_k$, is a $D\times D$ matrix. $i_1i_2...i_N$ label the input physical legs and $j_1j_2...j_N$ label the output physical legs. We are going to call the left and right legs of the $M^{j_ki_k}$ matrices the virtual legs and think of $M$ as a four leg tensor.

Pictorially, the local tensor $M$ in the MPO is given by
\begin{equation}
M = \raisebox{-.4\height}{\includegraphics[height=1.2cm]{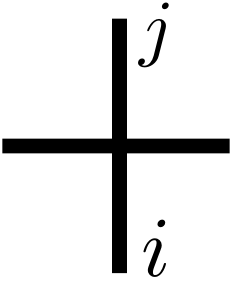}},
\end{equation}
while the total MPO is given by
\begin{equation}
O = \raisebox{-.4\height}{\includegraphics[height=1.2cm]{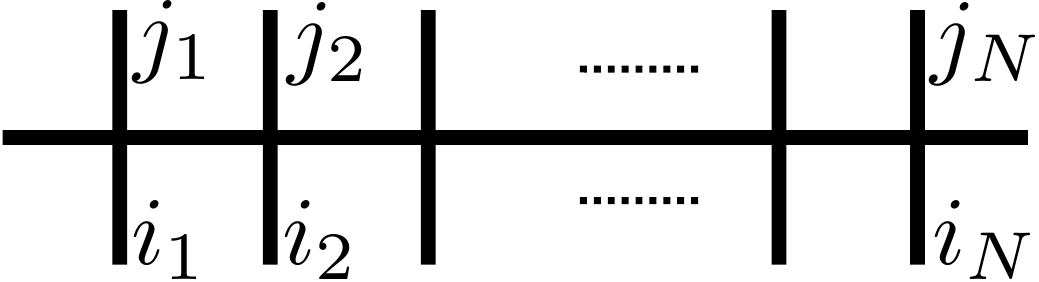}}.
\end{equation}

In order for $O$ to be unitary, it has to satisfy the condition that $O^{\dagger}O=I$. We consider the case where this is true for any finite system size, not just in the thermal dynamic limit. We call such operators Matrix Product Unitary Operators (MPUO).
\begin{mydef}[Matrix Product Unitary Operator]
Consider a matrix product operator $O$ represented with tensor $M$ of finite bond-dimension. $O$ is called a matrix product unitary operator if it is a unitary for all system sizes.
\label{def:MPUO}
\end{mydef}
Note that we emphasize `for all system sizes' for a good reason. In section~\ref{sec:beyondGNVW} we are going to see that there are matrix product operators which are unitary only for certain system sizes, and hence do not fit into this definition.

If we define a new tensor $M^{\dagger}$ as
\begin{equation}
{M^{\dagger}}^{ji} = \left(M^{ij}\right)^*
\end{equation}
Then the MPUO condition is given graphically as
\begin{equation}
O^{\dagger}O = \raisebox{-.4\height}{\includegraphics[height=1.2cm]{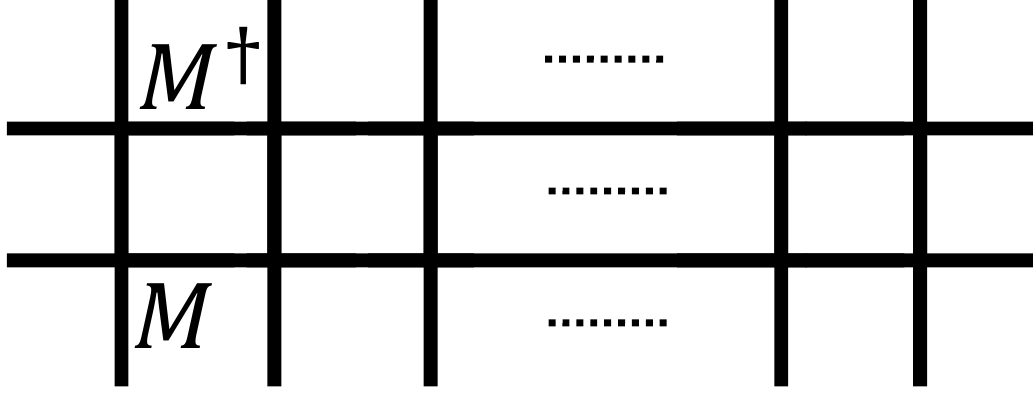}} = \raisebox{-.4\height}{\includegraphics[height=1.2cm]{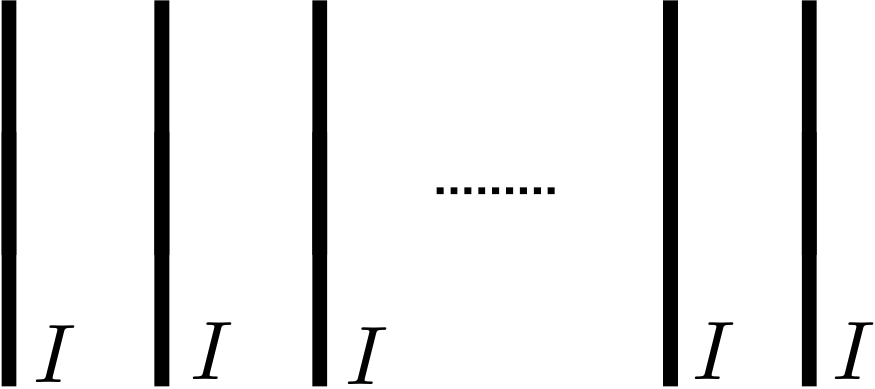}}
\label{uni}
\end{equation}
where we use a straight line to represent the identity matrix. This condition imposes very strong constraints on $M$. The constraint can be most easily identified on the composite of $M$ and $M^{\dagger}$ which we define as
\begin{equation}
T^{ij} = \sum_{k} {M^{\dagger}}^{ik}\otimes M^{kj} =\ \  \raisebox{-.4\height}{\includegraphics[height=1.3cm]{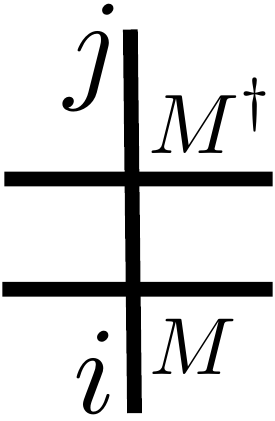}}
\label{Tij}
\end{equation}
The unitarity condition Eq.~\ref{uni} is saying that the matrix product operator with tensor $T^{ij}$ is equivalent to a tensor product of identity operators $I$ on each degree of freedom. If we combine the input and output physical legs of $T^{ij}$, we can think of it as representing a matrix product state, which would be a tensor product of maximally entangled pairs $|11\rangle+|22\rangle+ ... |dd\rangle$. 

Based on this observation, we can derive a general form for the $T^{ij}$ tensors. Let's give this as a lemma:

\begin{mylemm}\label{canonical_form}
Let $O$ be a Matrix Product Unitary Operator (MPUO) described by local tensor $M$, then the tensor $T^{ij}$, which is composed of $M$ and $M^{\dagger}$ as in Eq.~\ref{Tij}, has to take the following form:
\begin{align}\label{canonical_n}
\includegraphics[scale=0.15]{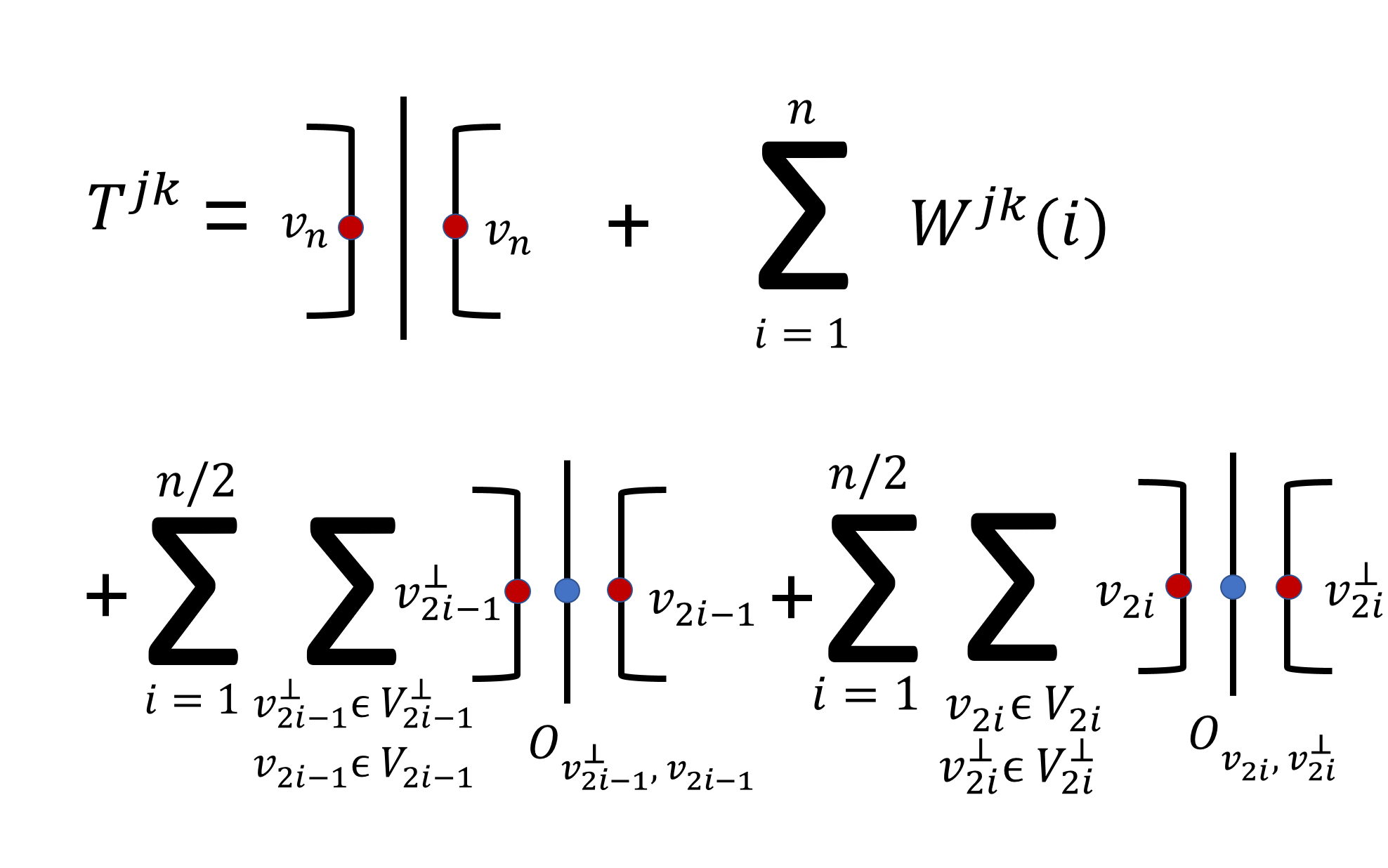}
\end{align}
where $n$ is a constant, which denotes the number of steps in the process of finding the canonical form of the associated MPS. $v_1, \ldots, v_n$ denotes vectors in the double virtual Hilbert space $V= \mathbb{C}^D \otimes \mathbb{C}^D$. Namely, each $v_i \in V_i$, $v^\perp_i \in V^\perp_i$ is an orthonormal basis vector in $V= V_n \oplus V^\perp_n \oplus V^\perp_{n-1} \oplus V^\perp_{n-2} \oplus \ldots \oplus V^\perp_{1}$ and $V_i= V_{i+1} \oplus V^\perp_{i+1}$ for all $0 \leq i \leq n-1$. Each $W^{jk}(i)$ denotes a block on $V^\perp_i$ which of similar form of $T^{jk}$ except $\langle v^\perp_i| W^{jk}_i |v^\perp_{i}\rangle=0$ for all $j,k$ and all $v^{\perp}_i \in V^{\perp}_i$.
\end{mylemm}

\begin{proof}
This form of the tensor $T$ follows directly from the definition of the canonical form given in Ref.\onlinecite{Perez-Garcia2007} and the requirement that $O$ is an MPO which is a unitary for all system sizes. We define an MPS form for the operator $O^{\dagger} O$ which is described by local tensor $A^{jk}$ obtained by combining the input and output legs, $j$ and $k$, of $T^{jk}$ as the physical legs, i.e.,

\begin{align}\label{T_to_A}
\includegraphics[scale=0.17]{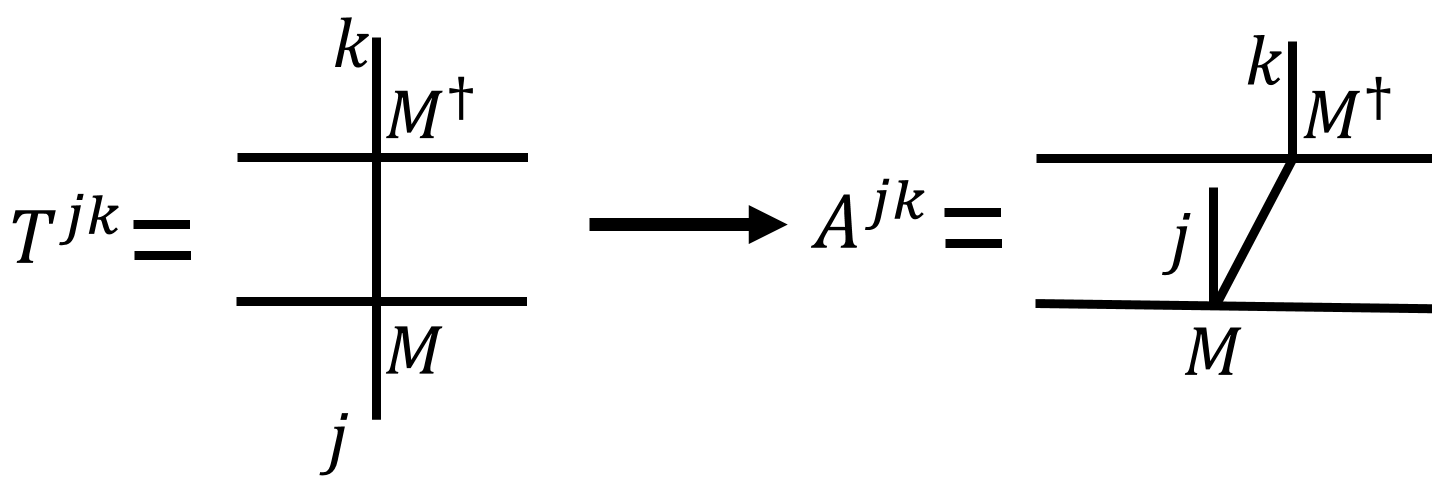}
\end{align}

Following the procedure of finding the canonical form given in Ref.~\cite{Perez-Garcia2007}, we step by step decompose the left and the right virtual vector space of the tensors $A^{jk}$ into orthogonal subspaces. The procedure does this alternatively, first $A^{ij}$ gets updated to $(P_{V_1}+ P_{V^\perp_1})A^{jk}$, where $P_{V_1}$ and $P_{V^{\perp}_{1}}$ are projectors onto $V_1$ and $V^{\perp}_1$ respectively, and $P_{V_1}+ P_{V^\perp_1}= P_{V_0}$ is the projector on the whole virtual space. As proved in Ref.~\cite{Perez-Garcia2007} the terms $P_{V_1}A^{jk}P_{V^\perp_1}$ vanishes. Now we update the MPS tensor $A^{jk}$ to $P_{V_1}A^{jk}P_{V_1}+ P_{V^{\perp}_1}A^{jk}P_{V_1}+ P_{V^{\perp}_1}A^{jk}P_{V^{\perp}_1}$. 

Repeating this procedure alternatively for decomposing left and right virtual vector spaces, we obtain the following general form of the tensor $T^{jk}$:

\begin{equation}\label{canonical_form_n}
\includegraphics[scale=0.15]{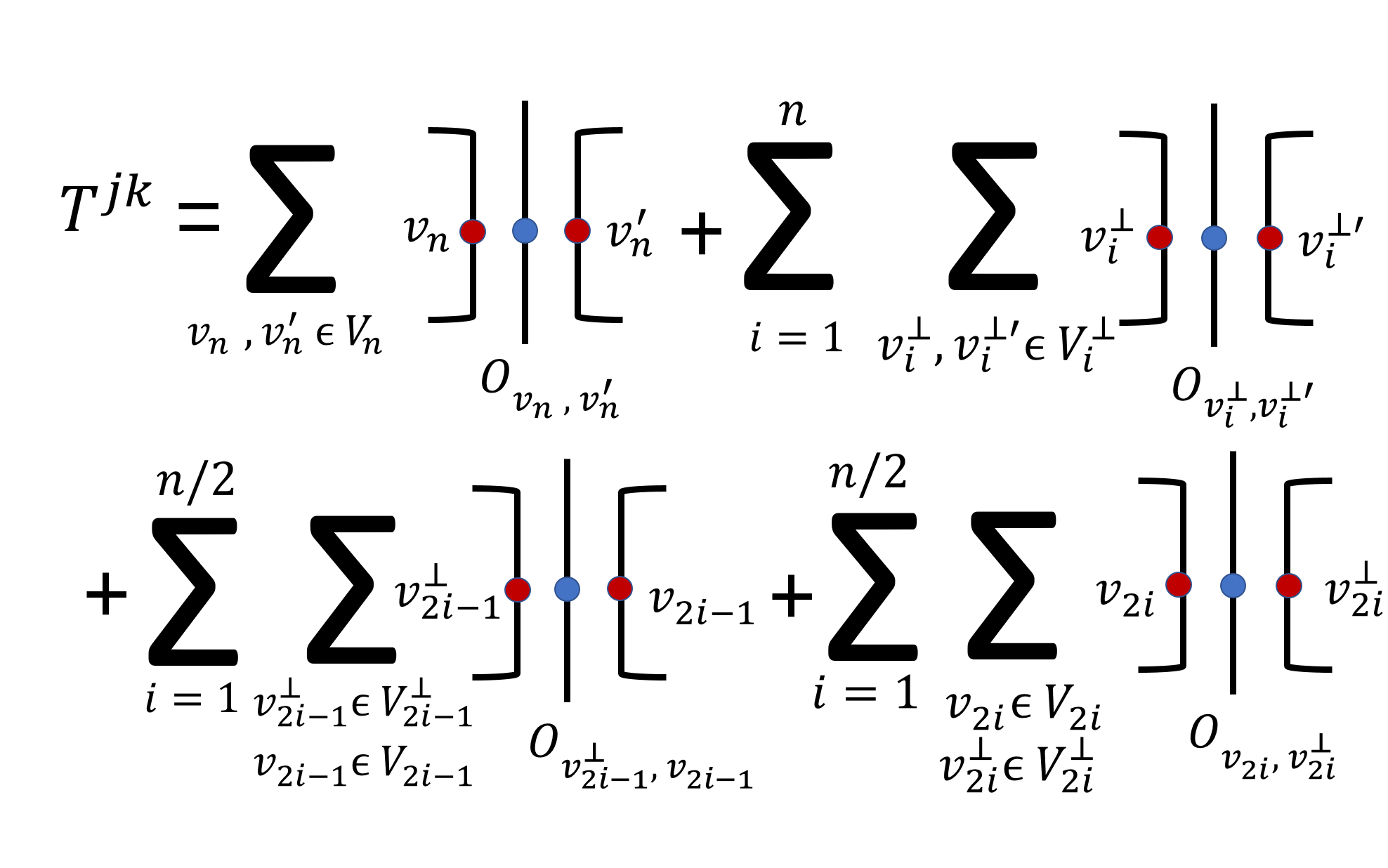}
\end{equation}

where the subspaces are split as $V= V_n \oplus V^{\perp}_n \oplus V^{\perp}_{n-1} \oplus \ldots V^{\perp}_{1}$, and $V_{i-1}= V_{i} \oplus V^{\perp}_{i}$ for all $1 \leq i \leq n$ and $V_0= V= \mathbb{C}^D \otimes \mathbb{C}^D$.

Now we impose the requirement that the MPO $O$ represented by the local tensor $M$ is unitary for all system sizes. Since the operator $O$ is obtained with periodic boundary conditions as seen in Eq.~\eqref{MPO}, we must investigate the associated MPS represented by local tensors $A^{ij}$ with periodic boundary conditions. This means that, only the operators of the form $O_{w,w'}$ with $w, w' \in V^{\perp}_i$ or $w, w' \in V_n$ appears in the expression of $O^{\dagger}O$. Since we know that $O^{\dagger}O= I^{\otimes N}$ for all system sizes, each of the operators $O_{w w'}$ must be individually equal to $I$. We can immediately see that only one block of these operators can have diagonal terms, since otherwise it would imply that $O^{\dagger} O$ is only proportional to $I^{\otimes N}$ and there is no way to make it exactly equal to $I^{\otimes N}$ by normalization. Let this block be the $n$th block that maps $V_n$ to $V_n$ from right to left virtual legs. This implies that in the general form of the MPS the blocks that map $V^{\perp}_i$ to $V^{\perp}_i$ can be decomposed further with the same procedure but without any diagonal term. Say all $\dim V^\perp_i=1$ for all $i \leq n$, then we only have diagonal term in the block that maps $V_n$ to $V_n$ from right to left virtual legs. When one of the blocks $V^\perp_i$ is such that $\dim V^\perp_i > 1$, we have additional terms in the expression of $A^{jk}$ that has only non diagonal terms in the the block $V^\perp_i$, which further decomposes as described above , but only within the vector space of $V^\perp_i$. We denote these terms in the sum as $W^{jk}(i)$ for each block of $V^\perp_i$. Furthermore, the fact that MPS is a product state means that $\dim V_n=1$. Hence,  Eq.~\eqref{canonical_form_n} and the fact that $O$ is an MPUO as defined in Def.~\ref{def:MPUO} imply that $T^{jk}$ is of the following form:

\begin{align}
\includegraphics[scale=0.15]{canonical_n}
\end{align}
\end{proof}

Note that $n \leq D^2 - 1$ which simply follows from dimension considerations.

\subsection{1D locality preserving unitaries as MPUO}
\label{sec:LP-MPUO}

In this section, we are going to show that all locality preserving 1D unitaries can be represented as MPUO.

Let's look at a few examples first and see how their representation fits the form in Lemma~\ref{canonical_form}.
\begin{itemize}
\item Example 1: Tensor product of unitary operators

This is a trivial case where the dimension of the virtual legs is $1$. Graphically, we denote it as
\begin{equation}
M_{\text{product}} = \raisebox{-.4\height}{\includegraphics[height=1.3cm]{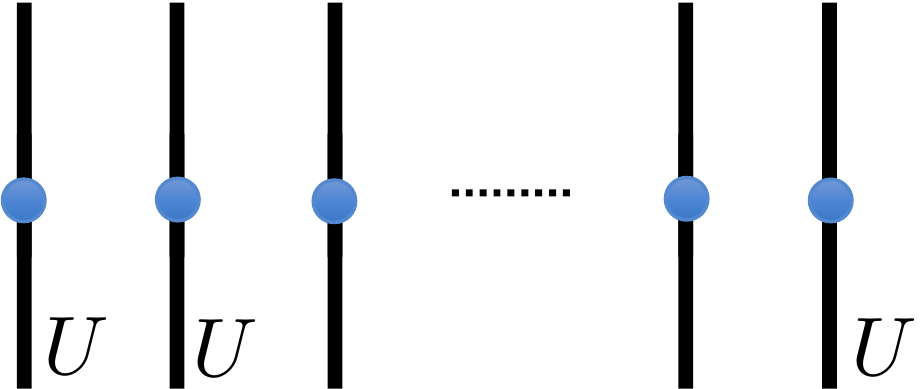}}
\end{equation}
where a line with a dot in the middle represents a nontrivial matrix, a unitary $U$ in this case.
The $T$ tensor as defined in Eq.~\ref{Tij} is automatically identity.

\item Example 2: Controlled-phase between nearest neighbor spin $1/2$s.

Let's consider a simple entangled unitary in 1D $\prod_{k=1}^N CP_{k,k+1}$, where each $CP_{k,k+1}$ is a two body unitary of the form
\begin{equation}
CP = \begin{pmatrix} 1 & 0 & 0 & 0 \\ 0 & 1 & 0 & 0\\ 0 & 0 & 1 & 0 \\ 0 & 0 & 0 & -1 \end{pmatrix}
\end{equation}
This unitary can be represented with
\begin{equation}
M_{CP} = \raisebox{-.4\height}{\includegraphics[height=1.5cm]{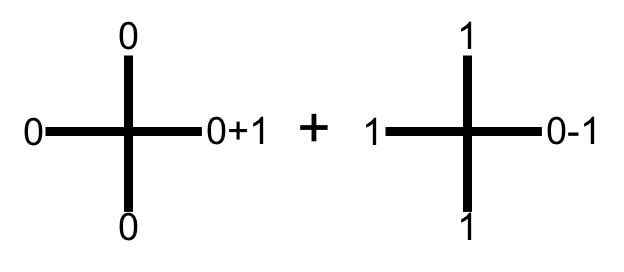}}
\label{MCP}
\end{equation}
We can check that $M_{CP}$ satisfies the condition in Lemma~\ref{canonical_form}. We can calculate $T_{CP}$ to be
\begin{equation}
\begin{array}{l}
T_{CP} = \raisebox{-.4\height}{\includegraphics[height=2cm]{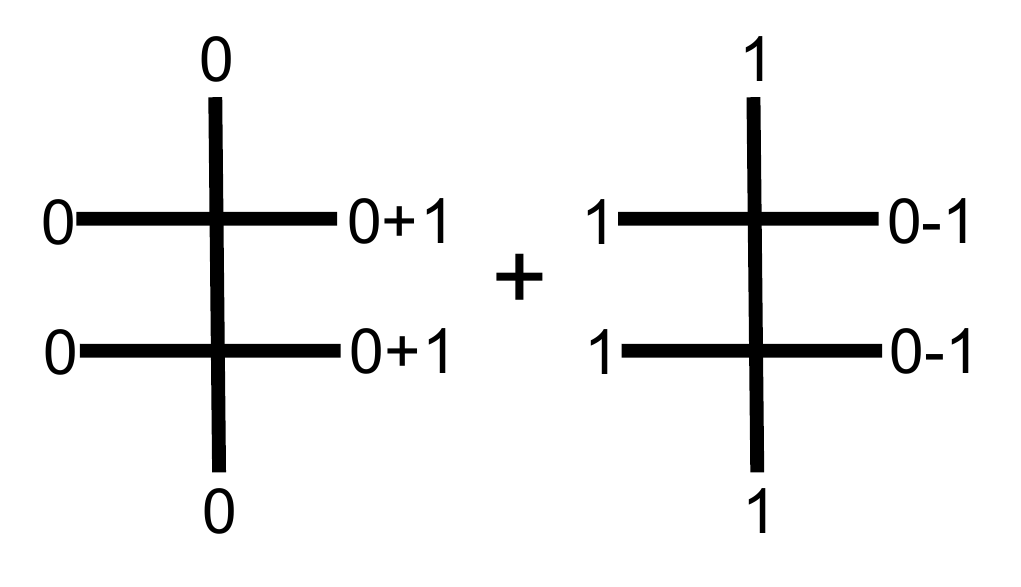}} = \\
\raisebox{-.4\height}{\includegraphics[height=1.5cm]{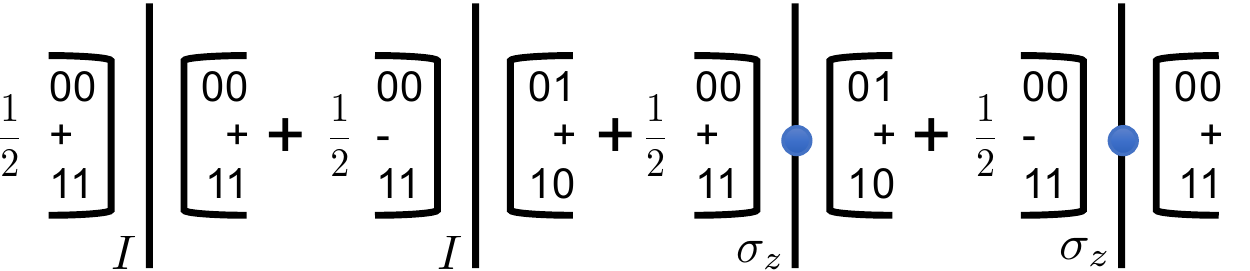}}
\end{array}
\end{equation}

\item Example 3: Translation

Translation, which is a locality preserving unitary that cannot be written as a finite depth circuit, can also be represented as a MPUO in a simple way. Consider the translation to the right by one step in a spin $1/2$ chain. The operator can be represented with
\begin{align}
M_r= \raisebox{-.4\height}{\includegraphics[height=1.2cm]{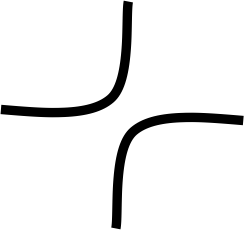}}
\label{M_r}
\end{align}
where the curved lines again represent the identity matrix between the left and up legs, and the right and down legs. When connected into a chain, it is straight forward to see that it represents translation.
\begin{align}
\raisebox{-.4\height}{\includegraphics[height=1.2cm]{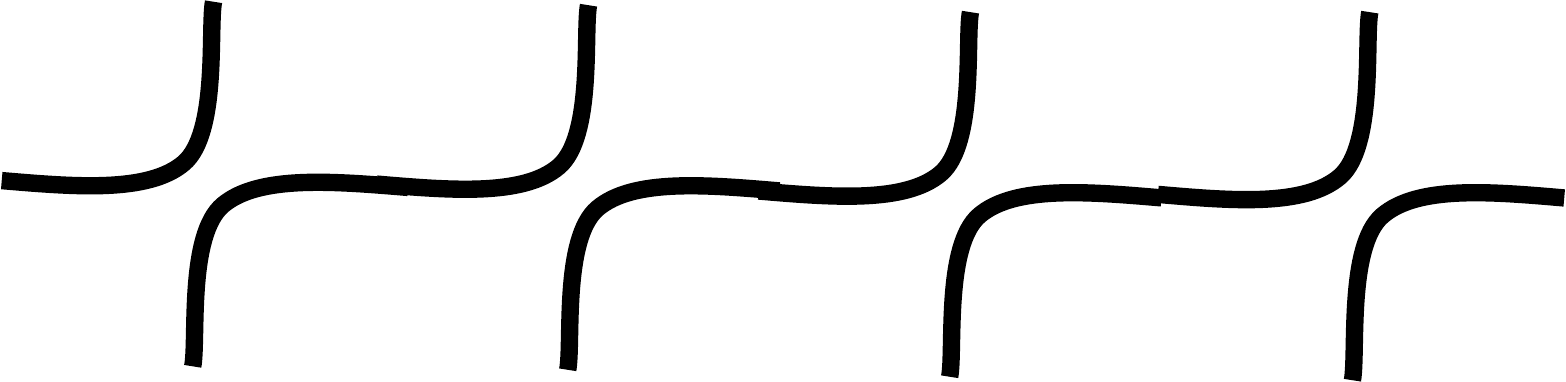}}
\end{align}
Similarly, translation to the left by one step can be represented with
\begin{align}
M_l= \raisebox{-.4\height}{\includegraphics[height=1.2cm]{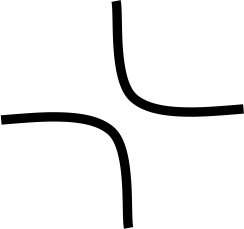}}
\end{align}
$M_r$ and $M_l$ also satisfy the condition in Lemma~\ref{canonical_form}. In particular,
\begin{equation}
\begin{array}{l}
T_r = \raisebox{-.4\height}{\includegraphics[height=1.2cm]{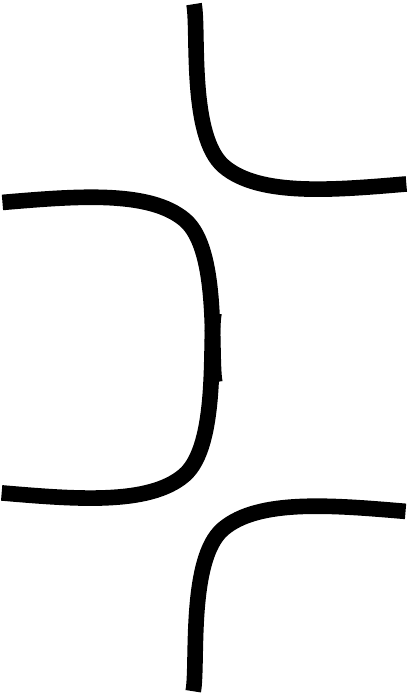}} = \\
\raisebox{-.4\height}{\includegraphics[height=1.5cm]{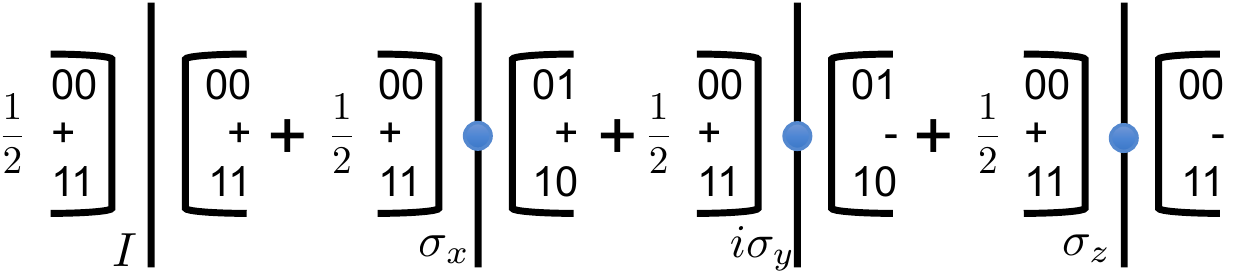}}
\end{array}
\end{equation}
And a similar expansion holds for $T_l$.
\end{itemize}

In fact, all locality preserving unitaries in 1D can be represented as MPUO satisfying Lemma~\ref{canonical_form}.

\begin{mythm}[Locality preserving 1D unitaries as MPUO]
Let $O$ be a locality preserving 1D unitary. It is possible to represent it as a Matrix Product Unitary Operator, as defined in Definition~\ref{def:MPUO}.
\label{LP-MPUO}
\end{mythm}
\begin{proof}
We prove this statement in the following steps:

1. Translation operator by one step can be represented with an MPO as shown with Example 3 above, such that the MPO is unitary for any system size.

2. One layer of non-overlapping unitaries can be represented with an MPO. WLOG, consider a layer of non-overlapping two-body unitaries, which can be represented with a tensor
\begin{align}
M_{tb} = \raisebox{-.4\height}{\includegraphics[height=1.2cm]{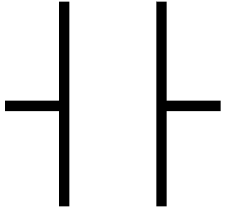}}
\label{Mtb}
\end{align}
when connected together into a chain, this tensor gives the two-body unitaries.
\begin{align}
\raisebox{-.4\height}{\includegraphics[height=1.2cm]{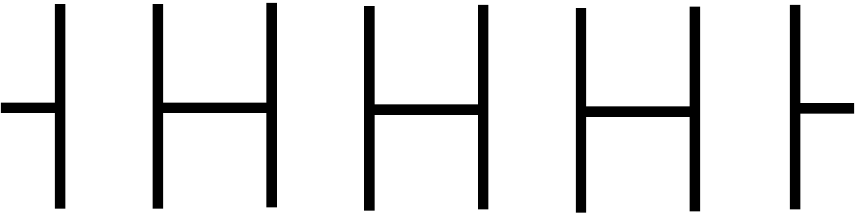}}
\label{Otb}
\end{align}
Such an MPO is unitary for all system sizes.

3. According to Ref.\onlinecite{Gross2012}, all 1D locality preserving unitaries can be decomposed into a finite number of layers of translation and finite depth local unitary circuits which can be further decomposed into a finite number of layers of non-overlapping few-body unitaries. The MPO representation of such a composite can be obtained by stacking the MPO representation for each component. As each component satisfies the MPUO condition that the MPO is unitary for all system sizes, the same is true for the composite MPO. Therefore, all 1D locality preserving unitaries can be represented as a MPUO, with tensors satisfying Lemma~\ref{canonical_form}.
\end{proof}

\section{\label{sec:conditions}Characterization of Matrix Product Unitary Operators}

In this section we prove fixed-point properties of MPUOs. Suppose that $O$ is an MPUO described by tensor $M$. We show that when the individual tensors are blocked, they satisfy equations that we call \emph{fixed-point equations}. These equations give a characterization of finite-bond dimension MPUOs. More importantly they imply that MPUOs are \emph{locality-preserving}. 

In order to obtain these results, we use basic facts about MPS\cite{Perez-Garcia2007}. So, let us first review these starting from the transfer matrix. Define the transfer matrix $\mathbb{E}_M$ of $M$ as
\begin{equation}\label{transfer_matrix}
\mathbb{E}_M = \sum_{ij} M^{ij} \otimes {M^{ij}}^* = \raisebox{-.4\height}{\includegraphics[height=1.3cm]{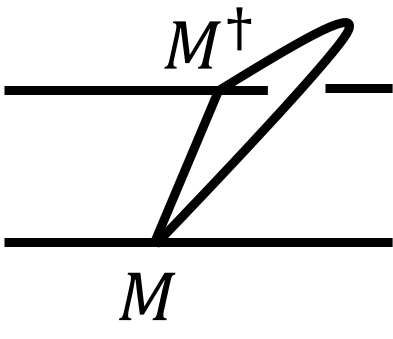}} = \sum_{i} T^{ii},
\end{equation}
and denote the right eigen-vector of $\mathbb{E}_M$ with largest eigenvalue as $r$ and the left eigen-vector with largest eigenvalue as $l$, such that $\langle l | r \rangle= 1$. Assuming the spectral radius of $\mathbb{E}$ is $1$, we have
\begin{equation}\label{left_right_eigenvectors}
\raisebox{-.4\height}{\includegraphics[height=3.8cm]{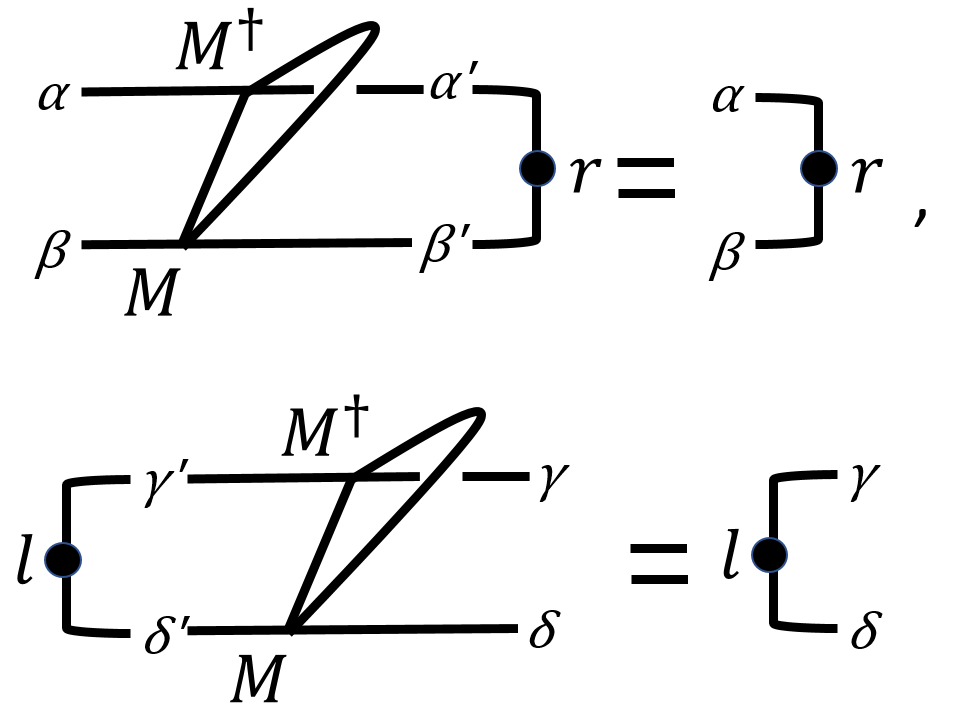}}
\end{equation}

Based on Lemma~\ref{canonical_form}, we can see that if $M$ describes an MPUO, the transfer matrix $\mathbb{E}_M$ is of the following form:

\begin{equation}
\begin{array}{ll}\label{transfer_M}
\mathbb{E}_M= &|v_n \rangle \langle v_n| + \displaystyle\sum^{n/2}_{i=1}\left(\displaystyle\sum_{v_{2i}, v^{\perp}_{2i}} tr(O_{v_{2i}, v^{\perp}_{2i}}) |v_{2i}\rangle \langle v^{\perp}_{2i}| \right.\\ 
& + \left. \displaystyle\sum_{v_{2i-1}, v^{\perp}_{2i-1}} tr(O_{v^{\perp}_{2i-1}, v_{2i-1}}) |v^{\perp}_{2i-1}\rangle \langle v_{2i-1}| \right)\\
&+ \displaystyle\sum^{n}_{i=1}\displaystyle\sum^{d}_{j=1} W^{jj}(i)
\end{array}
\end{equation}

Above we do not know the values of the trace of the operators, but we do know that the left and right eigen-vectors of $\mathbb{E}_M$ have to take the following form:

\begin{equation}\label{left_right_eigenvectors_EM}
\begin{array}{ll}
\langle l|&= \langle v_n| + \displaystyle\sum^{n/2}_{i=1} c_{2i-1} \langle v^{\perp}_{2i-1}|,\\
|r \rangle&= |v_n \rangle + \displaystyle\sum^{n/2}_{i=1} c_{2i} |v^{\perp}_{2i} \rangle
\end{array}
\end{equation}

where $c_i$s are complex coefficients.

The left and right eigenvectors, when seen as matrices $r= \sum_{\alpha \beta} r_{\alpha\beta} |\alpha\rangle\langle \beta|$ with elements $r_{\alpha \beta}$ and $l= \sum_{\gamma \delta} l_{\gamma\delta} |\gamma\rangle\langle \delta|$ with elements $l_{\gamma \delta}$, are positive matrices. $\langle l | r \rangle= 1$ since $\langle v_n | v_n \rangle=1$ and $\langle v^{\perp}_{2i}|v^{\perp}_{2j-1} \rangle=0$ for all $i,j$.

Now we are ready to state the results. We define $\tilde{M}^{JI}= M^{j_1 i_1}M^{j_2 i_2} \ldots M^{j_n i_n}$, where $I=i_1i_2...i_n$, $J=j_1j_2...j_n$, as the tensor obtained by blocking the individual tensor $M$. The blocked tensor $\tilde{M}$ satisfies the following fixed-point equations:

\begin{enumerate}

\item Fixed-point equation 1 - Separation:

\begin{align}\label{separation}
\raisebox{-.4\height}{\includegraphics[scale=0.2]{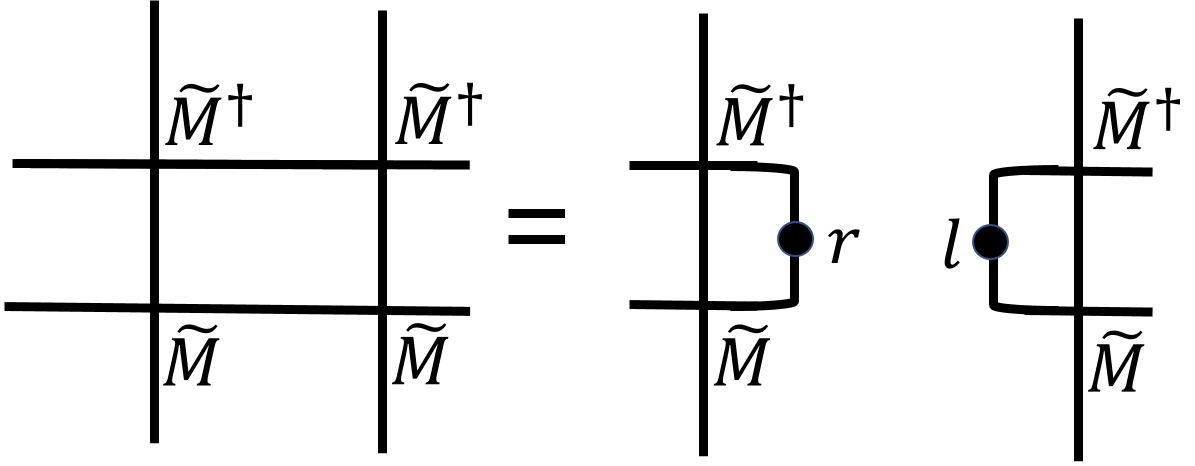}}
\end{align}

\item Fixed-point equation 2 - Isometry:

\begin{align}\label{isometry}
\raisebox{-.4\height}{\includegraphics[scale=0.20]{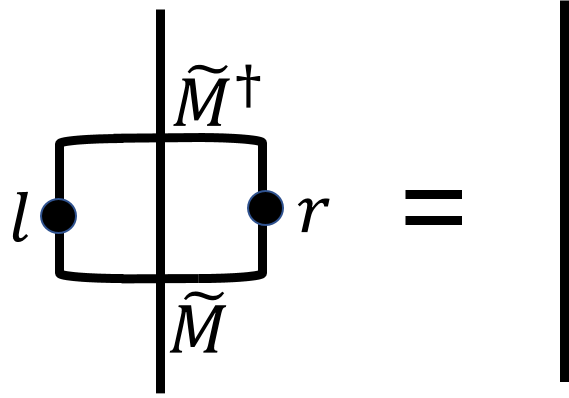}}
\end{align}

\end{enumerate}

where $l$ and $r$ denote the left and right eigenvectors of the transfer matrix $\mathbb{E}_M$ as given in Eq.~\eqref{left_right_eigenvectors_EM}. Eq.~\ref{separation} (separation) and Eq.~\eqref{isometry} (isometry) imply the following equations called \emph{pulling through} conditions, that we frequently make use of in the paper.

\begin{eqnarray} \label{pulling_through}
\nonumber \includegraphics[scale=0.2]{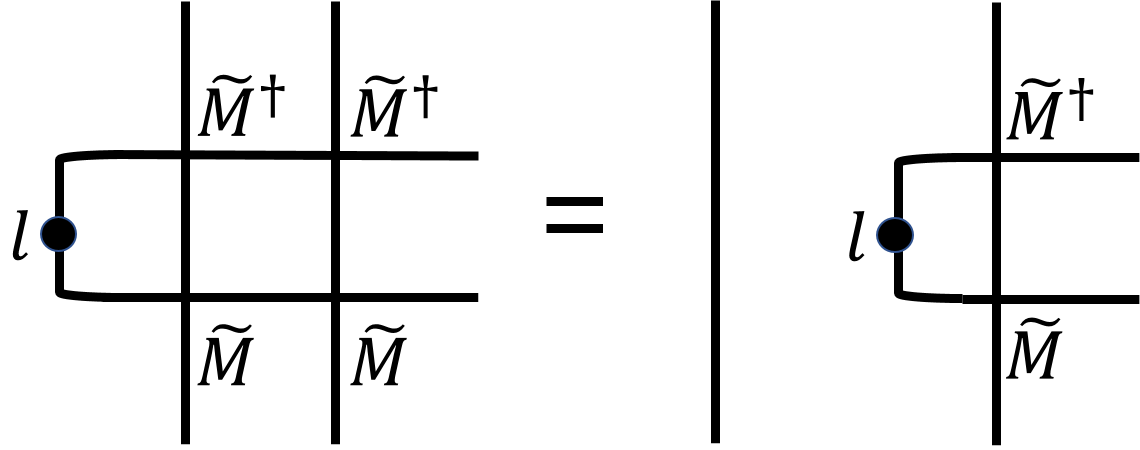}\\ 
\raisebox{-.9\height}{\includegraphics[scale=0.2]{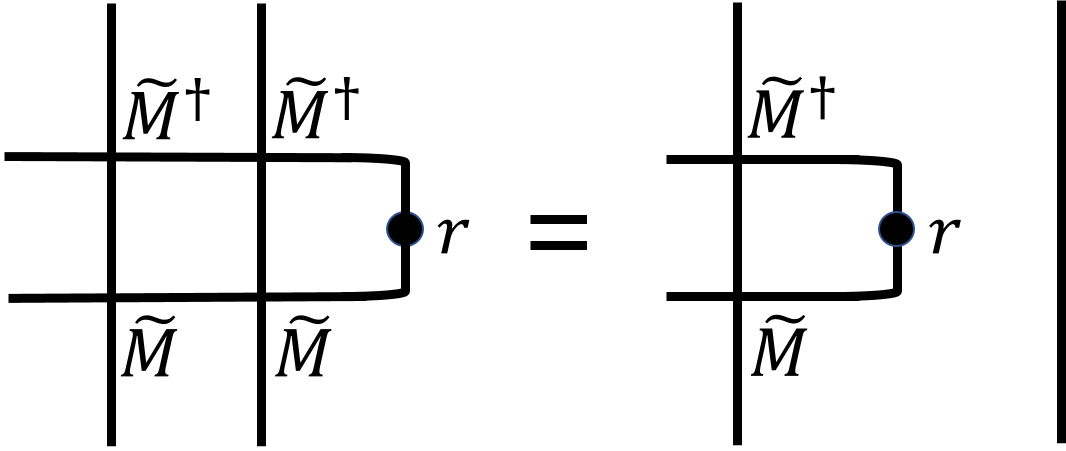} }
\end{eqnarray}

Before proving the above claims, we first give a lemma that explicitly shows the form of the tensor $\tilde{T}^{IJ}$ which is obtained by blocking the tensor $T^{ij}$ $D^2$-times, i.e., $\tilde{T}^{IJ}= T^{i_1 j_1} T^{i_2 j_2} \ldots T^{i_{D^2} j_{D^2}}$.

\begin{mylemm}\label{blocked_canonical_T}
Let the general form of the tensor $T$ be as in Eq.~\eqref{canonical_n} in Lemma~\ref{canonical_form}. Then, the blocked tensor $\tilde{T}^{IJ}=  T^{i_1 j_1} T^{i_2 j_2} \ldots T^{i_{D^2} j_{D^2}}$, where $D^2$ is the bond dimension of the tensor $T$, is of the following form
\begin{align}\label{blocked_canonical}
\raisebox{-.4\height}{\includegraphics[scale=0.14]{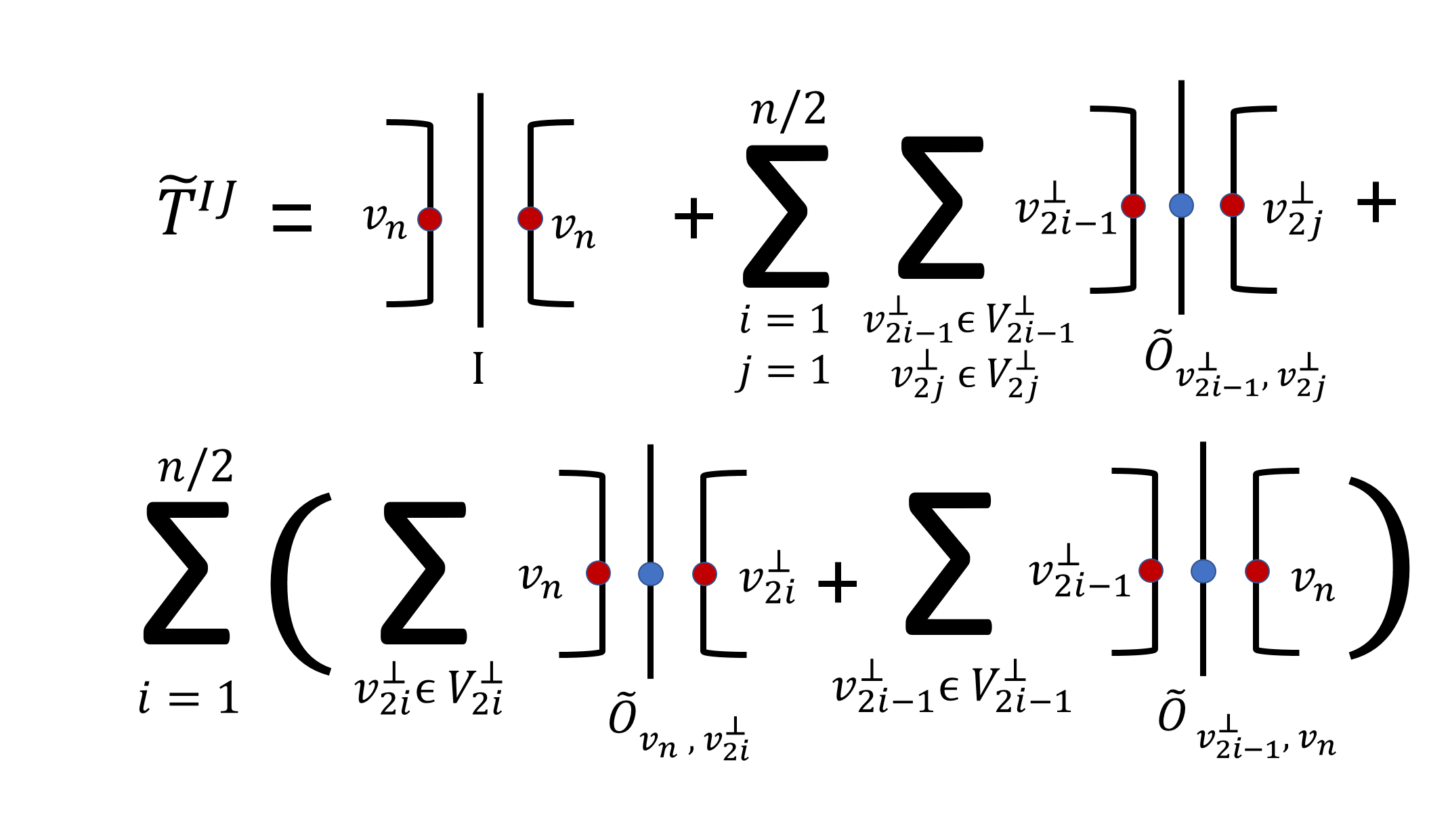}}
\end{align}
\end{mylemm}

\begin{proof}
By Lemma~\ref{canonical_form}, the general form of the tensor $T$ can be taken as 
\begin{equation}
\begin{array}{ll}
&T= |v_n \rangle I \langle v_n| + \displaystyle\sum^{n/2}_{i=1}\left(\displaystyle\sum_{v_{2i}, v^{\perp}_{2i}}  |v_{2i}\rangle O_{v_{2i}, v^{\perp}_{2i}} \langle v^{\perp}_{2i}| \right.\\ 
&+\left. \displaystyle\sum_{v_{2i-1}, v^{\perp}_{2i-1}} |v^{\perp}_{2i-1}\rangle O_{v^{\perp}_{2i-1}, v_{2i-1}} \langle v_{2i-1}| \right) + \displaystyle\sum_{i}W(i)
\end{array}
\end{equation}
where $|v_i\rangle \in V_i$, $|v^{\perp}_i \rangle \in V^\perp_i$, and $V= V_{n} \oplus V^\perp_{n} \oplus V^{\perp}_{n-1} \oplus \ldots \oplus V^{\perp}_{1}$. Now, imagine that we block $D^2$ of these tensors and obtain the tensor $\tilde{T}^{IJ}=  T^{i_1 j_1} T^{i_2 j_2} \ldots T^{i_{D^2} j_{D^2}}$. Using the facts that $\langle v^{\perp}_i |v^{\perp}_j \rangle=0$ for all $i$ and $j$, and  $\langle v_i|v^{\perp}_j \rangle= 0$ for all $j \leq i$, it's only a matter of careful book-keeping to show that only the terms with $|v_n \rangle \langle v_n|$, $|v^{\perp}_{2i-1} \rangle \langle v^{\perp}_{2j}|$, $|v_n \rangle \langle v^{\perp}_{2i}|$ and $|v^{\perp}_{2i-1} \rangle \langle v_n|$ appear in the expression of the tensor $\tilde{T}^{IJ}$. Notice that the $W(i)$ denotes the operator components within the block $V^\perp_{i}$ which only has nondiagonal elements and except the diagonal element it's in the same form as $T$. After blocking $D^2$ times, the terms that come from $W(i)$ do not contract anymore either from left or right, hence they don't appear in the tensor $\tilde{T}$. Note that the operator component with left and right indices $|v_n \rangle \langle v_n|$ acts as $I^{\otimes (D^2)}$. 
\end{proof}

Now, we prove that the blocked tensor $\tilde{M}$ that describes the MPUO satisfies the separation and isometry fixed-point equations given above in Eq.~\eqref{separation} and Eq.~\eqref{isometry}.

\begin{mythm}[MPUO implies fixed-point equations]
Let $O$ be an MPUO described by the tensor $M$. Then there exists a finite number $n$ such that the blocked tensor $\tilde{M}$, which is obtained by blocking $D^2$ of the tensor $M$, satisfies the fixed point equations, i.e., Eq.\eqref{separation} and Eq.\eqref{isometry}.
\label{MPU_implies_conditions}
\end{mythm}

\begin{proof}
By Lemma~\ref{canonical_form} and Lemma~\ref{blocked_canonical_T}, we know that an MPUO implies the general form for $\tilde{T}$ as in Eq.~\eqref{blocked_canonical}. By direct calculation the LHS of Eq.~\eqref{separation} is given as
\begin{equation}
\begin{array}{ll}
|v_n \rangle I \otimes I \langle v_n|& + \displaystyle\sum^{n/2}_{i,j} \displaystyle\sum_{v^{\perp}_{2i-1}, v^{\perp}_{2j}} |v^{\perp}_{2i-1} \rangle \tilde{O}_{v^{\perp}_{2i-1}, v_n} \otimes \tilde{O}_{v_n, v^{\perp}_{2j}} \langle v^{\perp}_{2j}|\\
&+ \displaystyle\sum^{n/2}_{i}\left( \displaystyle\sum_{v^{\perp}_{2i}} |v_n \rangle I \otimes \tilde{O}_{v_n, v^{\perp}_{2i}} \langle v^{\perp}_{2i}|\right.\\ 
&\left. +  \displaystyle\sum_{v^{\perp}_{2i-1}} |v^{\perp}_{2i-1} \rangle \tilde{O}_{v^{\perp}_{2i-1}, v_n} \otimes I \langle v_{n}|  \right)
\end{array}
\end{equation}
which is also equal to the RHS of the same equation, considering the fact that $\langle v_n | r \rangle = \langle l | v_n \rangle= 1$ and $\langle v^{\perp}_{2i}| r \rangle = \langle l | v^{\perp}_{2i-1}\rangle= 0$ for all $i$, which are easily seen from the form of the left and right eigen-vectors derived in Eq.~\eqref{left_right_eigenvectors_EM}. This concludes the proof of the separation equation. Using the same facts, it is straightforward to prove the isometry condition given in Eq.~\eqref{isometry}. It is the following equation that follows immediately from the above facts
\begin{align}
\langle l| \tilde{T} |r \rangle= I.
\end{align}
This completes the proof.
As a side remark it's also straightforward to see that the isometry equation~\eqref{isometry} is true even before blocking, i.e., $\langle l | T | r\rangle=I$.
\end{proof}

Theorem~\ref{MPU_implies_conditions} gives a characterization of MPUOs $O$ by tensors $\tilde{M}$ that satisfies the fixed-point equations, i.e., Eqs.\eqref{separation} and \eqref{isometry}.

Another consequence of the fixed-point equations is what we call the pulling through equations, which is given as a corollary as follows.

\begin{mycorol}
The fixed point equations, i.e., Eq.~\eqref{separation} and Eq.~\eqref{isometry}, imply the pulling through equations, i.e., Eq.~\eqref{pulling_through}.
\end{mycorol}

\begin{proof}
We start with the LHS of the pulling through equation, i.e., Eq.~\eqref{pulling_through}. We apply the fixed-point equations, namely separation, i.e., Eq.~\eqref{separation} and then apply the isometry, i.e., Eq.~\eqref{isometry}, respectively. Pictorially, it follows as below.
\begin{align}
\raisebox{-.4\height}{\includegraphics[scale=0.18]{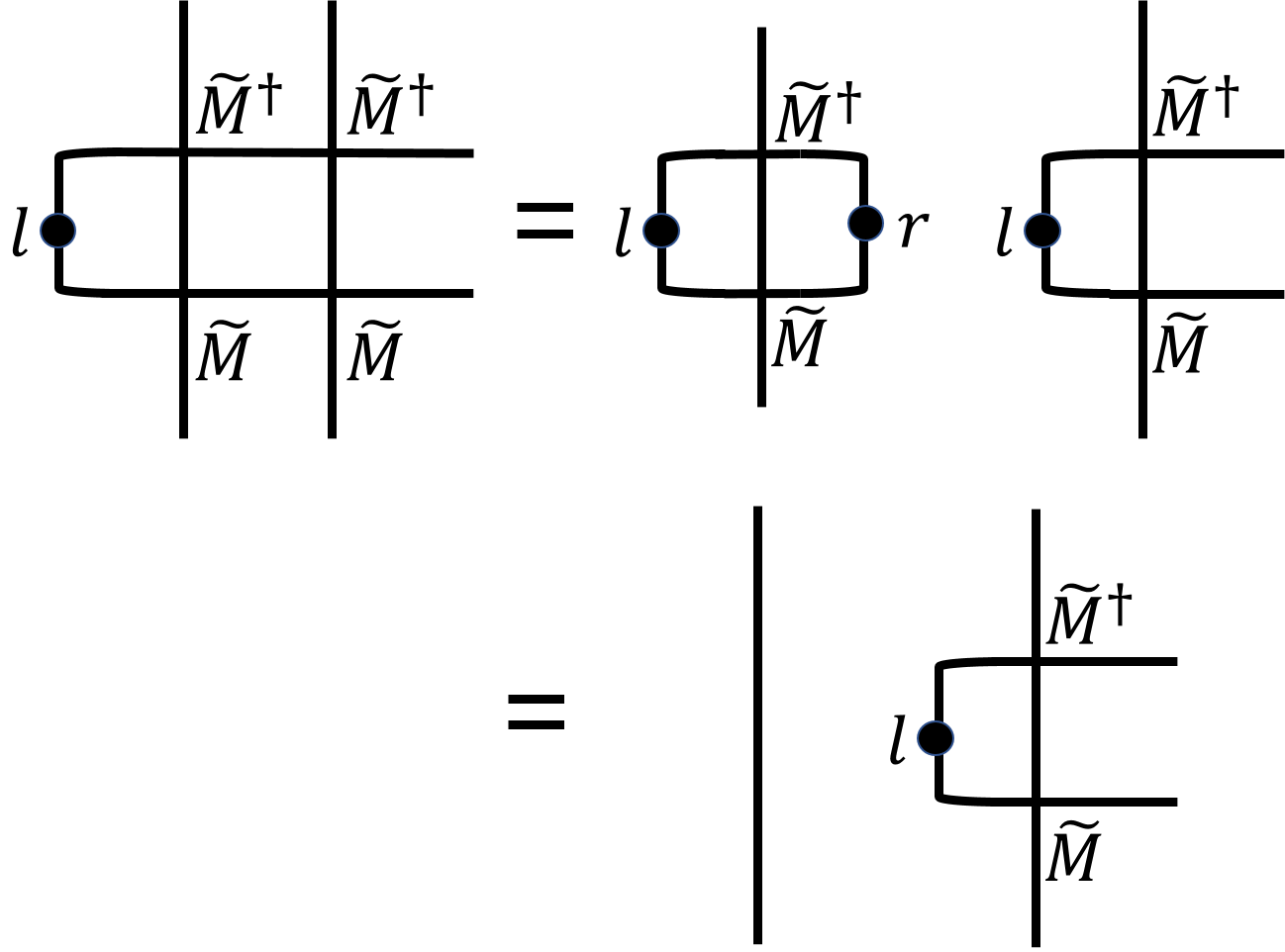}}
\end{align}
The other pulling through equation from right to left follows from separation and isometry fixed-point equations in the same way.
\end{proof}

Finally we close this section by showing that all finite-bond dimension MPUOs are locality-preserving. It means that, it maps any geometrically $k$-local operator to a geometrically $(k+c)$-local operator, where $c$ is a constant independent of the system size. This is proven in the following corollary.

\begin{mycorol}[MPUOs are locality-preserving]
Every MPUO is locality preserving, namely they map geometrically $k$-local operators to geometrically at most $(k+c)$-local operators where $c$ is a constant independent of the system size.
\end{mycorol}

\begin{proof}
An MPUO $O$ acts on an operator $O_k$ as $O: O_k \rightarrow O^{\dagger} O_k O$.Pictorially it is shown by 
\begin{align}
\raisebox{-.4\height}{\includegraphics[scale=0.18]{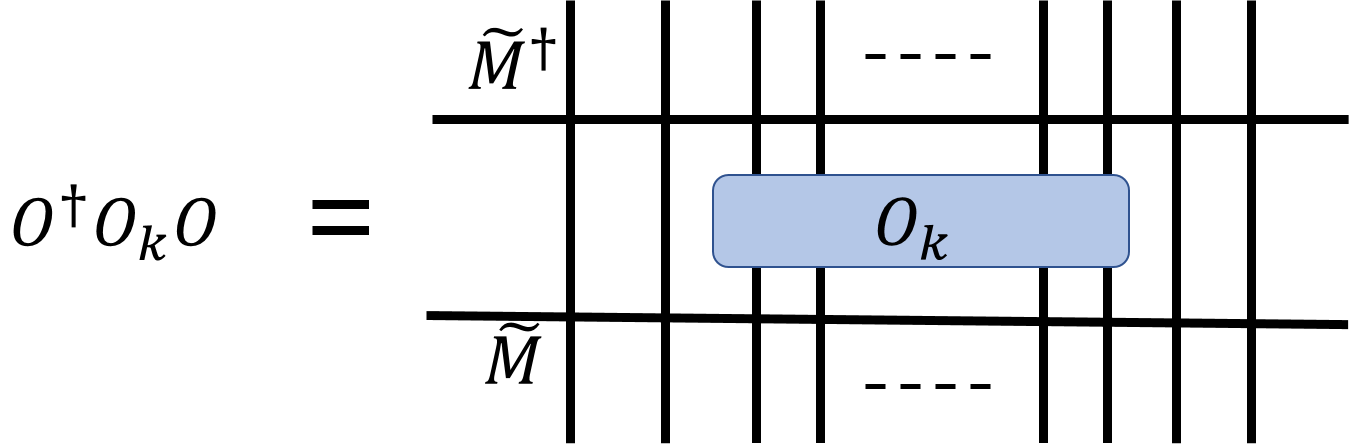}}
\end{align}
Using fixed-point equations, it is straightforward to see that 
\begin{align}
\raisebox{-.4\height}{\includegraphics[scale=0.18]{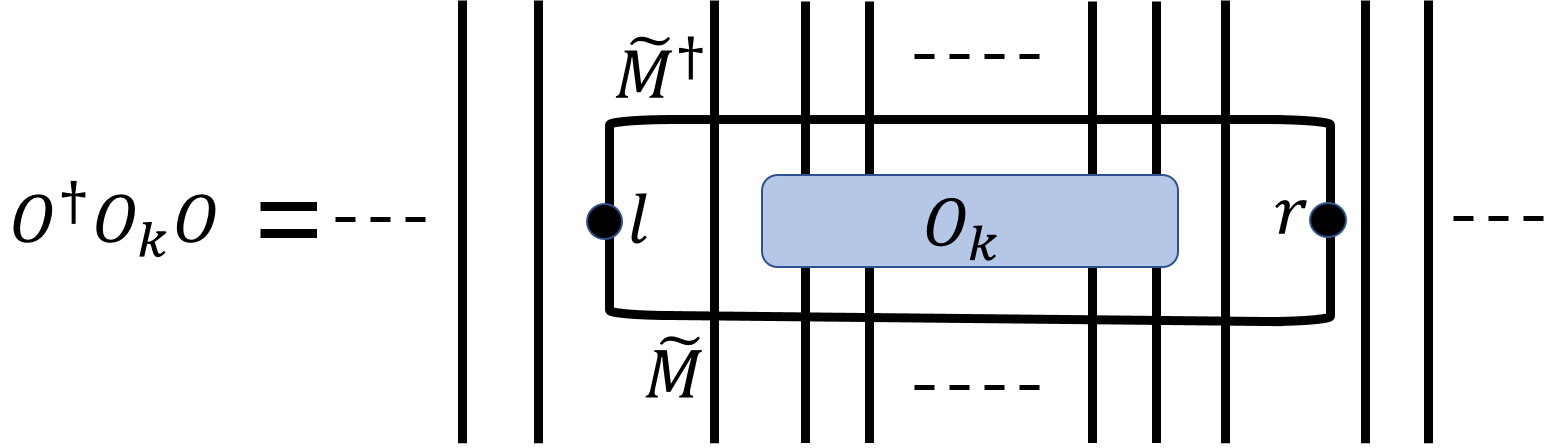}}
\end{align}
is a $(k+2)$-local operator. Hence after blocking sites, MPUOs map $k$-local operators to at most $(k+2)$-local operators. This means that before blocking, a $k$-local operator is mapped to at most a $(k+2D^2)$-local operator, since we are guaranteed to reach the fixed point after blocking $D^2$ sites.
\end{proof}

\section{\label{sec:index} Extracting GNVW index from MPUO representation}

\subsection{Review of GNVW index}
In Ref.\onlinecite{Gross2012}, Gross, Nesme, Vogts and Werner
proved that 1D locality preserving unitaries (called cellular automata in that paper) can be classified according to how much information is flowing across a cut in the chain. For example, finite depth local unitary circuits -- a finite number of layers of local unitaries where unitaries within each layer do not overlap with each other -- all belong to one class and there is zero information flow. On the other hand, translation by one step in a spin $1/2$ chain belongs to another class and there is a flow of a single spin $1/2$ across any cut. 

More specifically, Ref.\onlinecite{Gross2012} defined two 1D locality preserving untiaries to be equivalent to each other if and only if they differ from each other by a finite depth local unitary circuit and showed that every 1D locality preserving unitary is then equivalent to some translation operation. Each equivalence class is characterized by an index (the GNVW index) which measures how much translation is taking place: if there is a translation of $p$ dimensional Hilbert space by $m$ steps to the right, the index is $p^m$; if there is a translation of $q$ dimensional Hilbert space by $n$ steps to the left, the index is $1/q^n$; if there is translation in both directions, the index is $p^m/q^n$. Such an index is consistent with the equivalence class structure of locality preserving unitaries because it was shown that when two locality preserving operators multiply, their GNVW index also multiply
\begin{equation}
I_{\text{GNVW}}(O_1O_2) = I_{\text{GNVW}}(O_1)I_{\text{GNVW}}(O_2)
\label{III}
\end{equation}

For 1D locality preserving unitaries, the index is always a positive rational number and can be calculated as
\begin{equation}
I_{\text{GNVW}}(O) := \frac{\eta(O\mathcal{A}_LO^{\dagger},\mathcal{A}_R)}{\eta(\mathcal{A}_L,O\mathcal{A}_RO^{\dagger})}
\end{equation}
where $\mathcal{A}_L$ is the set of operators within distance $l_0$ on the left hand side of a cut and $\mathcal{A}_R$ is the set of operators within distance $l_0$ on the right hand side of the cut. $\eta(\mathcal{A},\mathcal{B})$ measures the overlap between the two sets of operators and is defined as
\begin{equation}
\eta(\mathcal{A},\mathcal{B}) := \frac{\sqrt{p_ap_b}}{p_{\Lambda}}\sqrt{\sum_{i,j=1}^{p_a}\sum_{l,m=1}^{p_b}\left|Tr_{\Lambda}\left(\hat{e}^{a\dagger}_{ij}\hat{e}^b_{lm}\right) \right|^2}
\end{equation}
where $\hat{e}^{a}_{ij}$ is the set of basis operators in $\mathcal{A}$ and there are $p_a$ of them; $\hat{e}^{b}_{lm}$ is the set of basis operators in $\mathcal{B}$ and there are $p_b$ of them; $\Lambda$ is a segment in the chain containing both $a$ and $b$. The GNVW index defined in this way converges to the positive rational number characterizing information flow when $l_0$ becomes large.

\subsection{Rank-ratio index $=$ $($GNVW index$)^2$}

How to extract the GNVW index from the matrix product representation of the locality preserving unitary operators? In this section, we show that it can be extracted as the square root of the Rank-Ratio index, which is defined as the ratio between the rank of the left and right SVD decompositions of the tensor $M$ in the representation.

\begin{mydef}[Rank-Ratio Index]
Let $M$ be the tensor in the matrix product representation of a unitary operator with physical legs in the up and down directions and virtual legs in the left and right directions. The Rank-Ratio Index is defined as the ratio between the rank of the SVD decomposition between left,down--right,up legs and the rank of the SVD decomposition between left,up--right,down legs. Graphically, the Rank-Ratio Index is given by
\begin{equation}
I_{\text{RR}}(M) = \text{rank}\left(\raisebox{-.4\height}{\includegraphics[height=0.8cm]{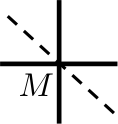}}\right) \bigg/ \text{rank}\left(\raisebox{-.4\height}{\includegraphics[height=0.8cm]{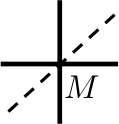}}\right)
\label{index}
\end{equation}
\end{mydef}

To demonstrate the connection between the Rank-Ratio index defined above and the GNVW index in Ref.~\onlinecite{Gross2012}, first we need to define the injectivity condition for matrix product operators. This definition is the same as the definition of injectivity as given in Ref.~\onlinecite{Perez-Garcia2007} if we combine the input and output physical legs of the MPO tensor and treat it as a matrix product state. We state this condition in detail below for subsequent discussions.

\begin{mydef}[Injective matrix product operator]
Consider a matrix product operator given by a set of matrices $\{M^{ij}\}$, where $i,j=1,..,d$ label the input and output physical legs. The MPO is called injective if $r_{\alpha\beta}$ and $l_{\gamma\delta}$ defined in Eq.~\ref{left_right_eigenvectors} are full rank matrices with row and column indices $\alpha$, $\beta$ and $\gamma$, $\delta$ respectively.
\label{injectivity}
\end{mydef}

The notion of injectivity is relevant to our discussion of MPUO because if $M$ represents an MPUO, then it can always be put into an injective form by removing redundant virtual leg dimensions. This can be shown by noticing that, if $M$ cannot be put into an injective form by removing redundant virtual dimensions, then each $M^{ij}$ contains at least two blocks in their canonical form. Then correspondingly $T^{ij}$ contains at least two blocks in its canonical form, which is not possible if it represents identity for all system sizes, as we argued below Eq.~\ref{canonical_form_n}.

Moreover, we need the following lemma.
\begin{mylemm}
Consider an MPO represented by an injective tensor $M$. Then
\begin{equation}
\lambda \left(\raisebox{-.4\height}{\includegraphics[height=1.8cm]{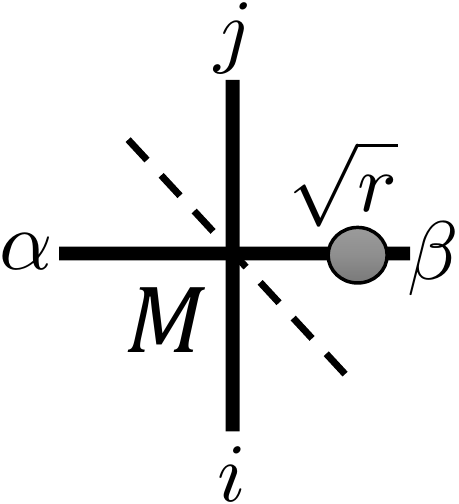}}\right) = \lambda^{1/2} \left(\raisebox{-.40\height}{\includegraphics[height=2.0cm]{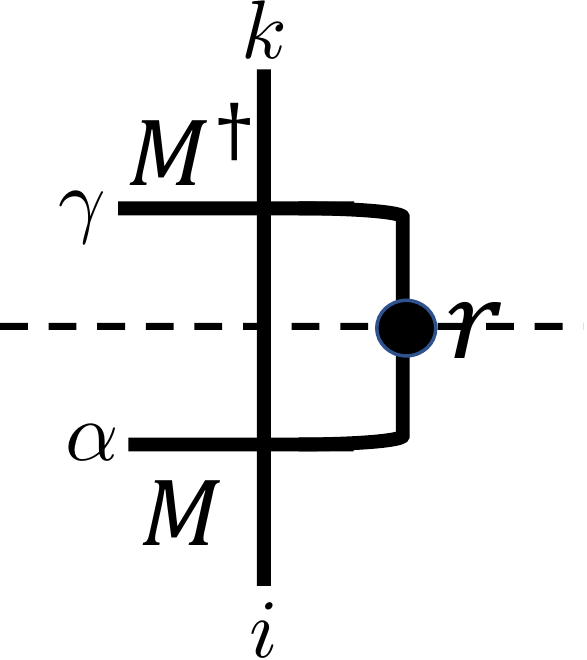}}\right)
\label{index_lemma_l}
\end{equation}
\begin{equation}
\lambda \left(\raisebox{-.4\height}{\includegraphics[height=1.8cm]{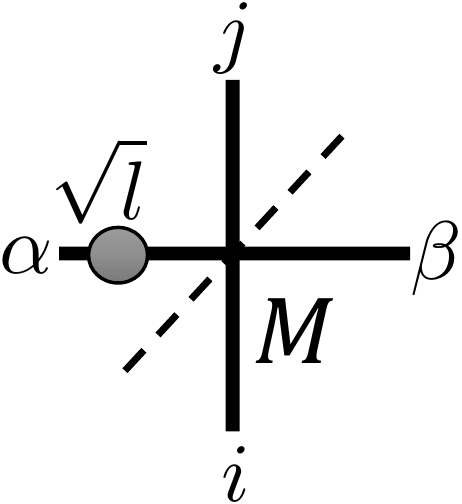}}\right) = \lambda^{1/2} \left(\raisebox{-.40\height}{\includegraphics[height=2.0cm]{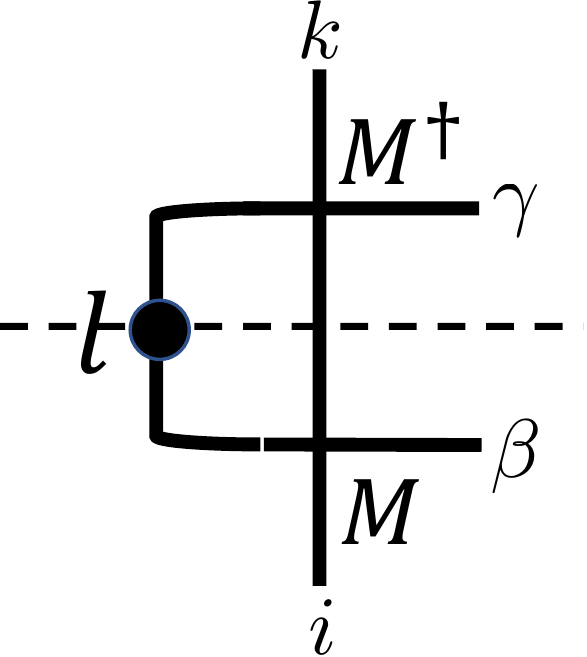}}\right) 
\label{index_lemma_r}
\end{equation}
where the dashed lines denote SVD decompositions across the cut, $\lambda$ denotes the set of singular values of the decomposition and the square root on $\lambda$ is taken element-wise. $l$ and $r$ are the left and right eigenvectors of the transfer matrix $\mathbb{E}_M$ as defined in Eq.~\ref{left_right_eigenvectors}. $l$ and $r$ are denoted with black dots and their square roots are denoted with grey dots.

Similarly, we have
\begin{equation}
\lambda \left(\raisebox{-.4\height}{\includegraphics[height=1.8cm]{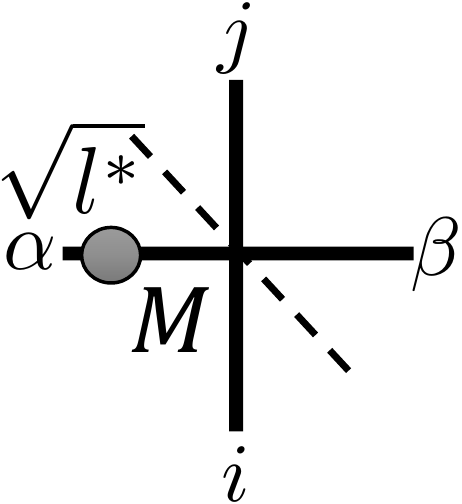}}\right) = \lambda^{1/2} \left(\raisebox{-.40\height}{\includegraphics[height=2.0cm]{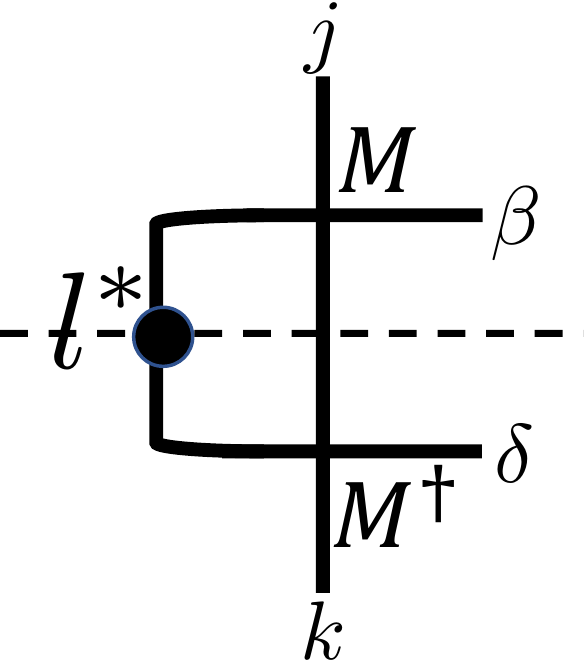}}\right)
\label{index_lemma_ls}
\end{equation}
\begin{equation}
\lambda \left(\raisebox{-.4\height}{\includegraphics[height=1.8cm]{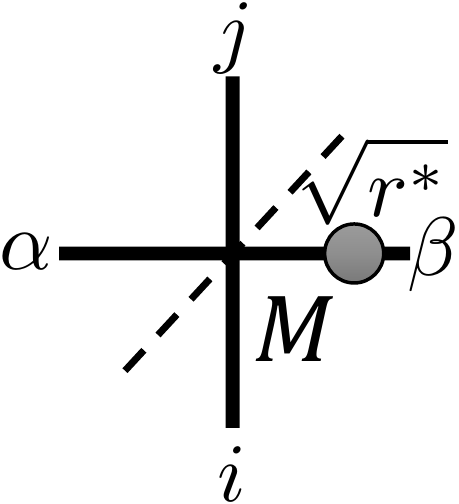}}\right) = \lambda^{1/2} \left(\raisebox{-.40\height}{\includegraphics[height=2.0cm]{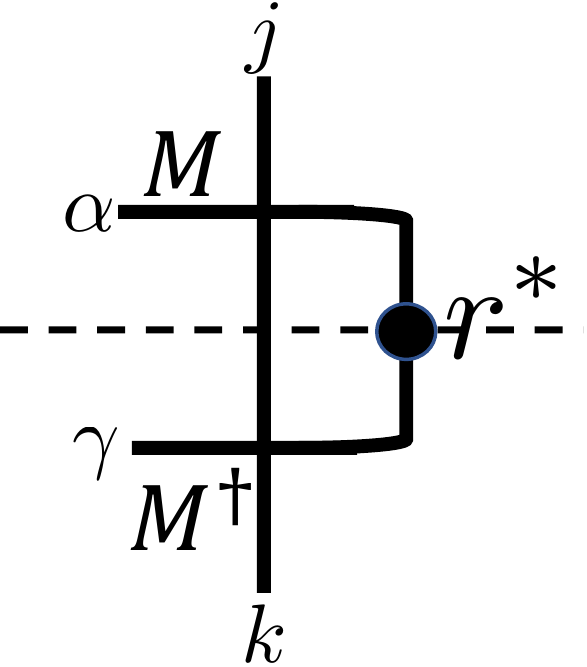}}\right) 
\label{index_lemma_rs}
\end{equation}
where $l^*$ is the complex conjugation of $l$ and $r^*$ is the complex conjugation of $r$.
\label{index_lemma}
\end{mylemm}
Note that as singular values are non-negative, so there is no ambiguity in taking the square root. Moreover, as $M$ is injective, $l$ and $r$ have full rank and have well defined square root.
\begin{proof}
We are going to prove Eq.~\ref{index_lemma_l} and then the proof of Eq.~\ref{index_lemma_r},~\ref{index_lemma_ls},~\ref{index_lemma_rs} is going to follow in a similar way.

Consider the SVD decomposition on the left hand side of Eq.~\ref{index_lemma_l} and suppose it takes the form
\begin{equation}
\sum_{\beta'} M_{i\alpha,j\beta'}\sqrt{r}_{\beta',\beta} = \sum_s U_{i\alpha,s}\lambda_{s}V_{s,j\beta}
\end{equation}
Then the tensor on the right hand side of Eq.~\ref{index_lemma_l} becomes
\begin{equation}
\begin{array}{ll}
 & \sum_{\beta',\delta',j} M_{i\alpha,j\beta'}r_{\beta',\delta'}M^{\dagger}_{j\delta',k\gamma}  \\
= & \sum_{j,s,s'} U_{i\alpha,s}\lambda_{s}V_{s,j\beta}V^{\dagger}_{s',j\beta}\lambda_{s'}U^{\dagger}_{k\gamma,s'} \\
= & \sum_{s} U_{i\alpha,s} \lambda^2_{s} U^{\dagger}_{k\gamma,s}
\end{array}
\end{equation}
Therefore, the singular value for the tensor on the right hand side is the square of the singular value on the left hand side. Hence we get Eq.~\ref{index_lemma_l}.
\end{proof}

The Rank-Ratio Index defined above can be directly related to the GNVW index if the MPUO is either injective or a stack of injective MPUOs.

\begin{mythm}[Rank-Ratio index = $($GNVW index$)^2$ for injective or stack of injective MPUO]\label{thm:RR_GNVW}
Consider an MPUO $O$ represented with tensor $M$. Take a sufficiently long but finite block so that the blocked tensor $\tilde{M}$ satisfies the Separation,  Isometry and Pulling Through conditions in Eq.~\ref{separation}, ~\ref{isometry} and ~\ref{pulling_through}. If $M$ is injective, or a stack of several injective tensors as $\left( \raisebox{-.4\height}{\includegraphics[height=1.0cm]{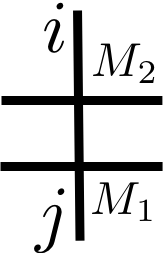}}\right)$, then
\begin{equation}
I_{\text{RR}}(\tilde{M}) = (I_{\text{GNVW}}(O))^2
\label{RR_GNVW}
\end{equation}
\end{mythm}

We are going to proceed to prove theorem~\ref{thm:RR_GNVW} in the following steps:
\begin{enumerate}
\item For an injective MPUO representation of non-overlapping two-body unitaries, 
\begin{equation}
I_{\text{RR}}(\tilde{M})=1=I^2_{\text{GNVW}}(O).
\end{equation}
\item For an injective MPUO representation of translation (to the right) by one step, \begin{equation}
I_{\text{RR}}(\tilde{M})=d^2=I^2_{\text{GNVW}}(O).
\end{equation} 
where $d$ is the dimension of the local physical Hilbert space.
\item If we stack two injective MPUOs as $M_{12} = \raisebox{-.4\height}{\includegraphics[height=1.0cm]{stack}}$, then 
\begin{equation}
I_{\text{RR}}(\tilde{M}_{12})=I_{\text{RR}}(\tilde{M}_1)I_{\text{RR}}(\tilde{M}_2).
\end{equation}
According to Ref.\onlinecite{Gross2012}, any locality preserving unitary can be obtained by stacking translation and layers of non-overlapping few body unitaries and their GNVW index multiply when stacked. Therefore, using the above equations we can show that the Rank-Ratio index of the stacked tensor is the square of the GNVW index. 
\item On the other hand, the stacked $M_{12}$ may not be injective itself but can be made injective. We will show that its Rank-Ratio index does not change even if we reduce it to the injective form. 
\item Finally, we show that the Rank Ratio index is stable in that if $\tilde{M}$ is the fixed point form (which satisfies Eq.~\ref{separation},~\ref{isometry} and ~\ref{pulling_through}) of an injective tensor or a stack of injective tensors, then the Rank-Ratio index does not change if we keep blocking $\tilde{M}$.
\end{enumerate}

\begin{proof}
Let's follow the procedure listed above.

1.  Consider the tensor given in Eq.~\ref{Mtb} to represent non-overlapping two-body unitaries. 
\begin{align}
M_{tb} = \raisebox{-.4\height}{\includegraphics[height=1.2cm]{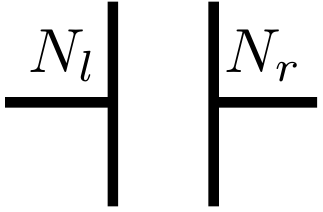}}
\label{Mtb2}
\end{align}
where we have labeled the left and right part of the tensor $N_l$ and $N_r$ respectively. As this representation can be obtained by decomposing each two-body unitary into a matrix product form, we can always choose $M_{tb}$ to be injective. 

According to the isometry condition in Eq.~\ref{isometry}, which is true even before blocking, we have
\begin{align}
\raisebox{-.4\height}{\includegraphics[height=1.5cm]{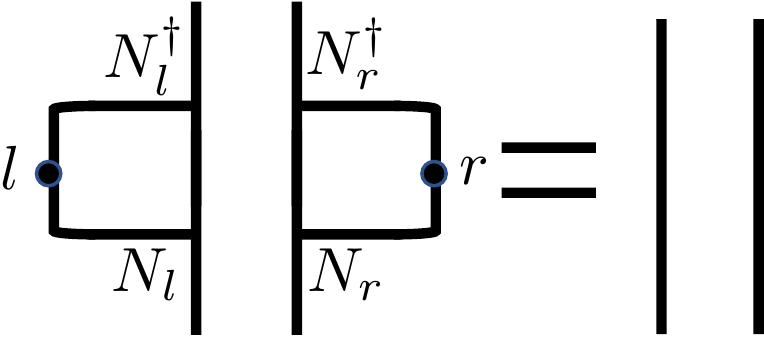}}
\label{EMtb}
\end{align}
and similarly
\begin{align}
\raisebox{-.4\height}{\includegraphics[height=1.5cm]{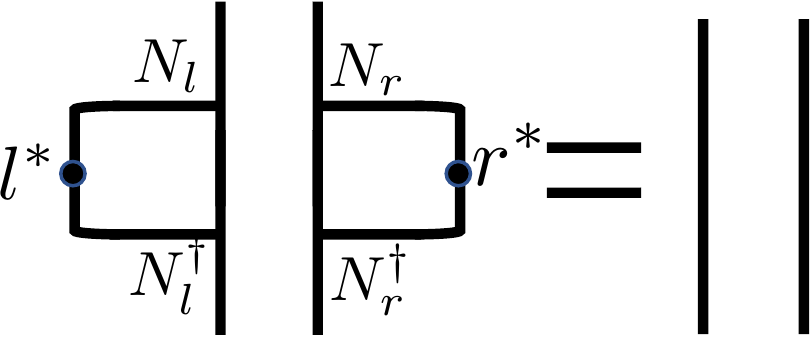}}
\label{EMtbs}
\end{align}
Each of these two equations actually contains two parts: the left halves on the two sides are equal to each other and right halves on the two sides are equal to each other. Both halves have to be satisfied simultaneously. Then using Eq.~\ref{index_lemma_l}, we have
\begin{equation}
\lambda\left(\raisebox{-.4\height}{\includegraphics[height=1.5cm]{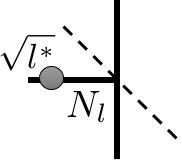}} \right) = \lambda^{1/2} \left(\raisebox{-.4\height}{\includegraphics[height=1.5cm]{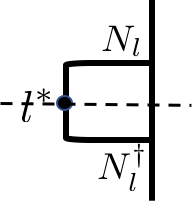}} \right) =  \lambda^{1/2} \left(\raisebox{-.4\height}{\includegraphics[height=1.5cm]{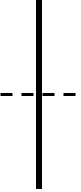}} \right)
\label{ltb}
\end{equation}
As $\sqrt{l^*}$ is a positive matrix, applying it does not change the rank of the SVD decomposition, so we have
\begin{equation}
\text{rank}\left(\raisebox{-.4\height}{\includegraphics[height=1cm]{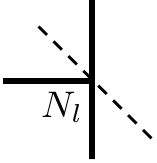}}\right) = \text{rank} \left(\raisebox{-.4\height}{\includegraphics[height=1cm]{ltb3}}\right) = d_l
\label{rtb}
\end{equation}
where $d_l$ is the dimension of the physical index in $N_l$. Similarly, we have
\begin{equation}
\begin{array}{l}
\text{rank}\left(\raisebox{-.4\height}{\includegraphics[height=1.0cm]{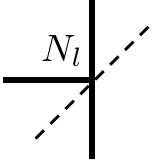}}\right)=d_l, \\
\text{rank}\left(\raisebox{-.4\height}{\includegraphics[height=1.0cm]{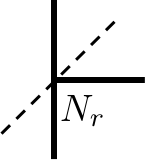}}\right)=\text{rank}\left(\raisebox{-.4\height}{\includegraphics[height=1.0cm]{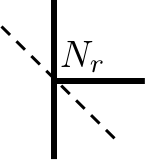}}\right)=d_r.  
\end{array}
\end{equation}
where $d_r$ is the dimension of the physical index in $N_r$. Now if we calculate the Rank-Ratio index for $M_{tb}$, we find that
\begin{equation}
\begin{array}{ll}
 & I_{RR}(M_{tb}) \\
 = & \text{rank}\left(\raisebox{-.4\height}{\includegraphics[height=0.8cm]{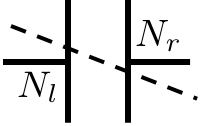}}\right) \bigg/ \text{rank}\left(\raisebox{-.4\height}{\includegraphics[height=0.8cm]{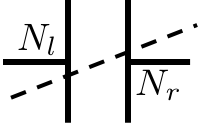}}\right) \\
= & (d_ld_r)/(d_ld_r) =1 = I^2_{\text{GNVW}}(O_{tb})
\end{array}
\end{equation}
Moreover, since the tensor $M_{tb}$ already satisfies the separation, isometry, pulling-through conditions in  Eq.~\ref{separation}, ~\ref{isometry} and ~\ref{pulling_through}, we have
\begin{equation}
I_{\text{RR}}(\tilde{M}_{tb})=1=I^2_{\text{GNVW}}(O_{tb}).
\end{equation}

2. For translation operator, the relation between the Rank-Ratio index and the GNVW index can be found through direct calculation. Consider translation by one step to the right represented by $M_r$ in Eq.~\ref{M_r}.
\begin{equation}
\begin{array}{ll}
 & I_{RR}(M_{r}) \\
 = & \text{rank}\left(\raisebox{-.4\height}{\includegraphics[height=0.8cm]{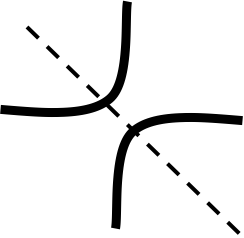}}\right) \bigg/ \text{rank}\left(\raisebox{-.4\height}{\includegraphics[height=0.8cm]{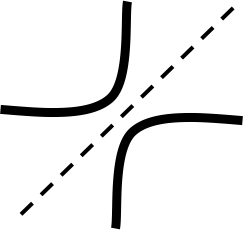}}\right) \\
= & d^2/1= d^2 = I^2_{\text{GNVW}}(O_r)
\end{array}
\end{equation}
Since the tensor $M_{r}$ already satisfies the fixed point conditions, we have
\begin{equation}
I_{\text{RR}}(\tilde{M}_{r})=d^2=I^2_{\text{GNVW}}(O_{r}).
\end{equation}
Moreover, even though we have only checked this relation for one possible representation of the translation operator, it holds for all possible injective representations as they differ from each other at most by a basis transformation on the virtual legs\cite{Perez-Garcia2007}.

3.  Now let us stack two layers of MPUOs which are injective individually. The composite tensor
\begin{equation}
M_{12} = \raisebox{-.4\height}{\includegraphics[height=1.0cm]{stack}}
\end{equation}
is in general not injective. But we will show that its Rank-Ratio index is still the square of the GNVW index of the corresponding unitary operator. 

Let's assume that $M_1$ and $M_2$ are already at fixed point form satisfying the separation, isometry, pulling through conditions Eq.~\ref{separation}, ~\ref{isometry}, ~\ref{pulling_through}. $M_{12}$ is in general not in a fixed point form, but by blocking sites we can take it to a fixed point form. Suppose that the fixed point for $M_{12}$ can be achieved by blocking two sites. (Our proof below also works if we take larger blocks.) Now we are going to use Eq.~\ref{index_lemma_l} through ~\ref{index_lemma_rs} in Lemma~\ref{index_lemma} to prove that
\begin{equation}
I_{\text{RR}}(\tilde{M}_{12}) = I_{\text{RR}}(\tilde{M_1})I_{\text{RR}}(\tilde{M_2}).
\end{equation}
To see this, we find that
\begin{equation}
\begin{array}{ll}
& \lambda\left(\raisebox{-.4\height}{\includegraphics[height=1.2cm]{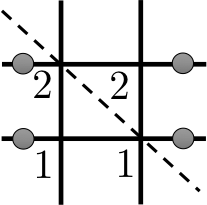}} \right)  = \lambda^{1/2}\left(\raisebox{-.4\height}{\includegraphics[height=1.6cm]{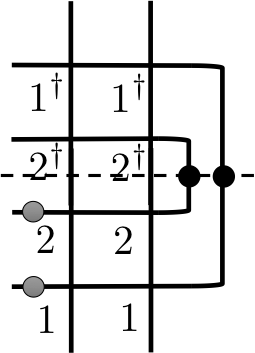}} \right) = \lambda^{1/2}\left(\raisebox{-.4\height}{\includegraphics[height=1.6cm]{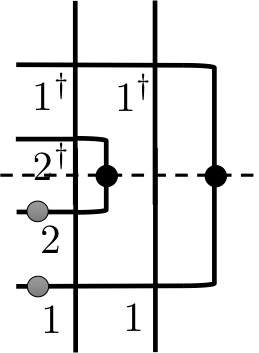}} \right)\\
= & \lambda\left(\raisebox{-.4\height}{\includegraphics[height=1.2cm]{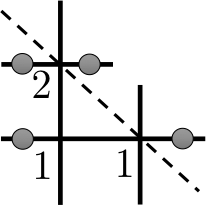}} \right)  = \lambda^{1/2}\left(\raisebox{-.4\height}{\includegraphics[height=1.6cm]{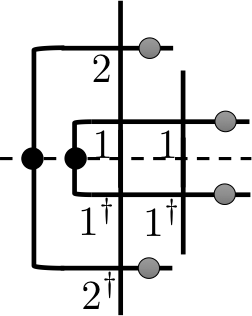}} \right)  = \lambda^{1/2}\left(\raisebox{-.4\height}{\includegraphics[height=1.6cm]{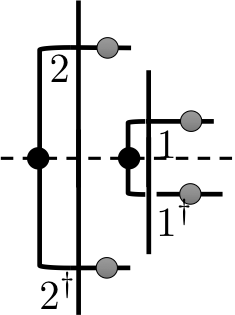}} \right)\\
= & \lambda\left(\raisebox{-.4\height}{\includegraphics[height=1.2cm]{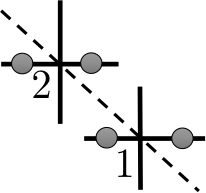}} \right)
\end{array}
\label{I12}
\end{equation}
where we have used simplified notation $1$, $2$, $1^{\dagger}$, $2^{\dagger}$ to refer to $M_1$, $M_2$, $M_1^{\dagger}$ and $M_2^{\dagger}$. The black dots represent the left and right eigenvectors of the transfer matrices of $M_1$, $M_2$, $M_1^{\dagger}$ and $M_2^{\dagger}$ while the grey dots are the square root of the black dots. As long as $M_1$, $M_2$ are injective (so are $M_1^{\dagger}$ and $M_2^{\dagger}$), the grey dots do not change the rank of the SVD decomposition. Therefore we have
\begin{equation}
\text{rank}\left(\raisebox{-.4\height}{\includegraphics[height=1.2cm]{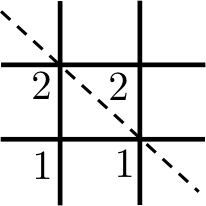}}\right) = \text{rank}\left(\raisebox{-.4\height}{\includegraphics[height=0.7cm]{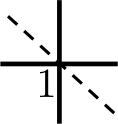}}\right)\text{rank}\left(\raisebox{-.4\height}{\includegraphics[height=0.7cm]{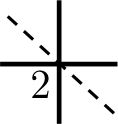}}\right)
\label{r12l}
\end{equation}
Similarly, we have
\begin{equation}
\text{rank}\left(\raisebox{-.4\height}{\includegraphics[height=1.2cm]{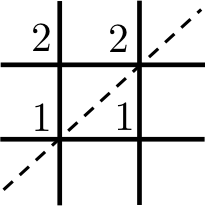}}\right) = \text{rank}\left(\raisebox{-.4\height}{\includegraphics[height=0.7cm]{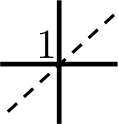}}\right)\text{rank}\left(\raisebox{-.4\height}{\includegraphics[height=0.7cm]{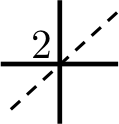}}\right)
\label{r12r}
\end{equation}

Dividing these two equations, we get as promised
\begin{equation}
I_{\text{RR}}(\tilde{M}_{12}) = I_{\text{RR}}(\tilde{M}_1)I_{\text{RR}}(\tilde{M}_2).
\end{equation}

4. $M_{12}$ as a stack of $M_1$ and $M_2$ may not be injective itself. But as we show below, the Rank-Ratio index does not change if we reduce it to the injective form. Suppose that to reduce $M_{12}$ to the injective form and remove redundant virtual dimensions, we need to do a projection $P$ to the pair of virtual legs in each direction, as denoted by the $\{$ $\}$ in the following equation
\begin{equation}
M_{12} = \raisebox{-.4\height}{\includegraphics[height=1.2cm]{stack}} \rightarrow \mathcal{M}_{12} = \raisebox{-.4\height}{\includegraphics[height=1.2cm]{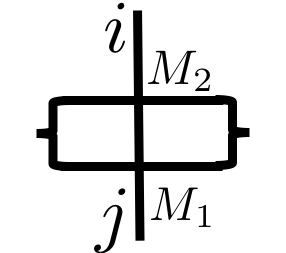}}
\end{equation}
The separation condition on $M_{12}$ reads
\begin{equation}
\raisebox{-.4\height}{\includegraphics[height=2.1cm]{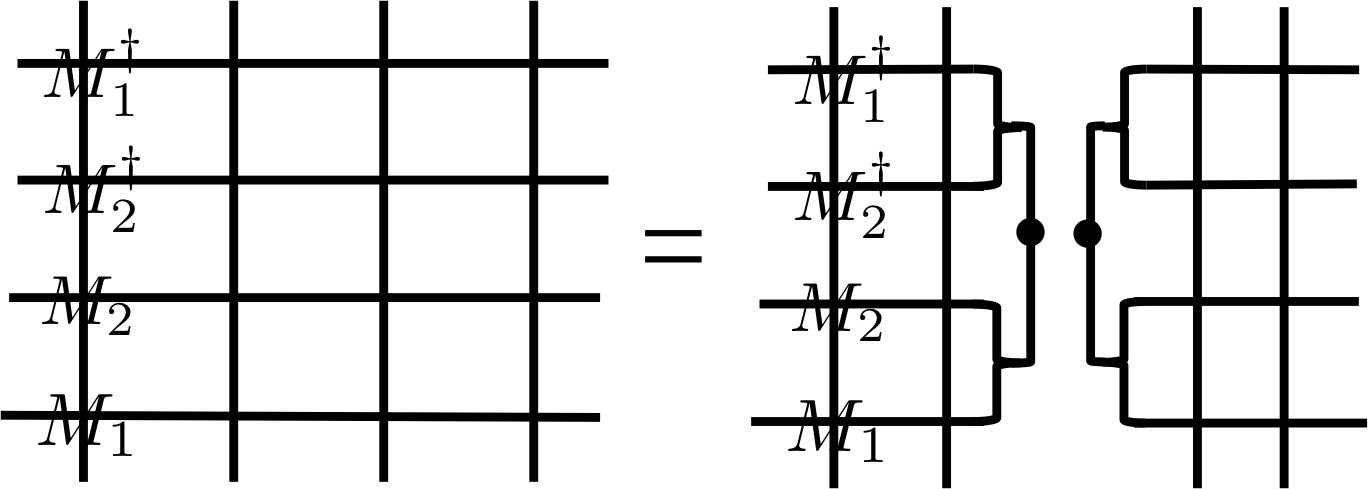}}
\label{M12Sep1}
\end{equation}
The left and right eigen-vectors (the black dots) that we insert in the middle are supported on $P$, so we are free to add those projections. (Note that the separation condition holds even if the MPO is not injective.) The tensors in each row are the same and we have labeled only one of them.

On the other hand, we have
\begin{equation}
\raisebox{-.4\height}{\includegraphics[height=4.5cm]{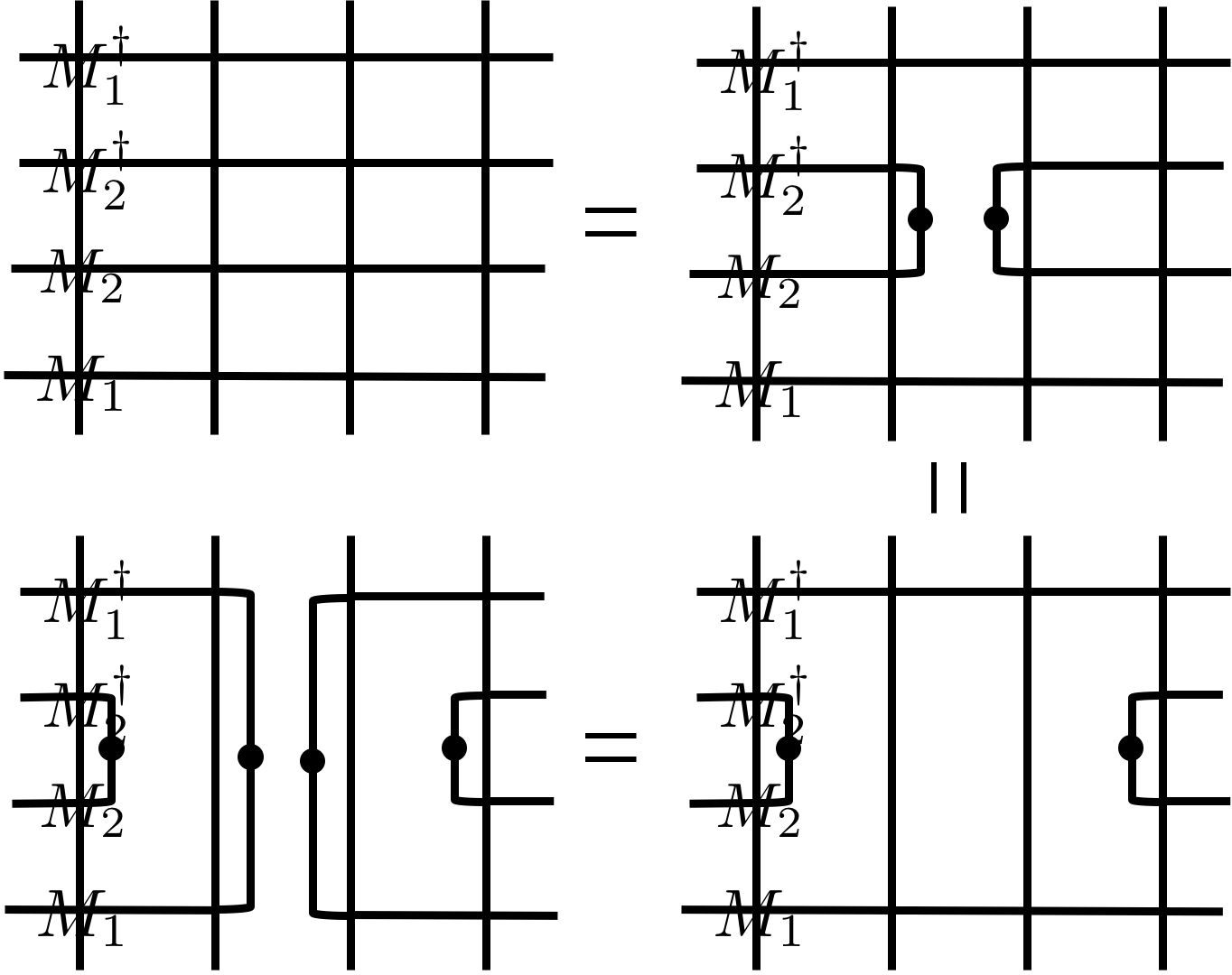}}
\label{M12Sep}
\end{equation}
where the first step uses the separation condition for $M_2$, the second step uses the pulling through condition for $M_2$, and the third step uses the separation condition for $M_1$.
Comparing Eq.~\ref{M12Sep1} and ~\ref{M12Sep}, we find that
\begin{equation}
\raisebox{-.4\height}{\includegraphics[height=2cm]{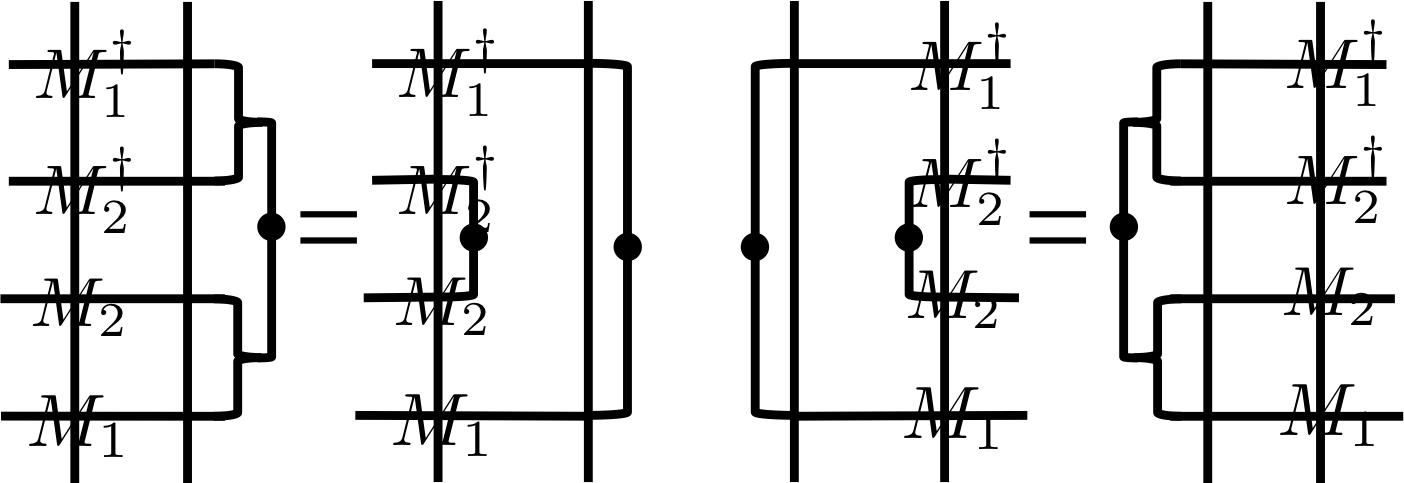}}
\label{M12Sep2}
\end{equation}
And a similar relation holds if we switch the place of $M_1M_2$ and $M^{\dagger}_2M^{\dagger}_1$.
From these relations we find that 
\begin{equation}
\begin{array}{ll}
& \lambda\left(\raisebox{-.4\height}{\includegraphics[height=1.0cm]{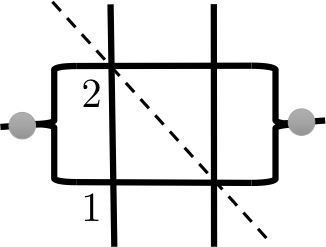}} \right)  = \lambda^{1/2}\left(\raisebox{-.4\height}{\includegraphics[height=1.6cm]{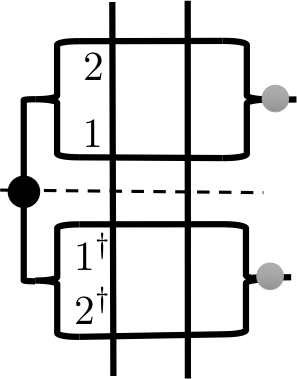}} \right) = \lambda^{1/2}\left(\raisebox{-.4\height}{\includegraphics[height=1.6cm]{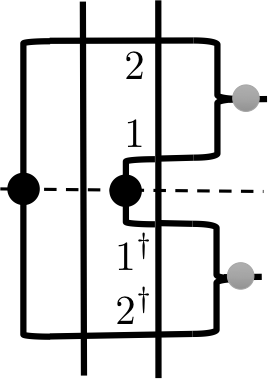}} \right)\\
= & \lambda\left(\raisebox{-.4\height}{\includegraphics[height=1.0cm]{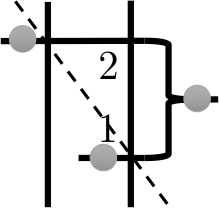}} \right)  = \lambda\left(\raisebox{-.4\height}{\includegraphics[height=1.0cm]{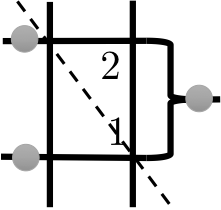}} \right)  = \lambda^{1/2}\left(\raisebox{-.4\height}{\includegraphics[height=1.6cm]{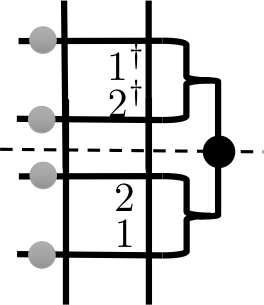}} \right)\\
= & \lambda^{1/2}\left(\raisebox{-.4\height}{\includegraphics[height=1.6cm]{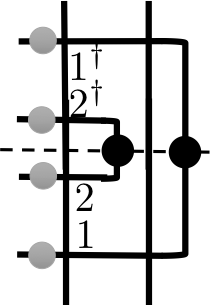}} \right) = \lambda\left(\raisebox{-.4\height}{\includegraphics[height=1.0cm]{I124}} \right) =  \lambda\left(\raisebox{-.4\height}{\includegraphics[height=1.0cm]{I127}} \right)
\end{array}
\label{I12_inj}
\end{equation}
In this equation, step 1, 3, 5, 7 uses Lemma~\ref{index_lemma}, step 2, 6 uses Eq.~\ref{M12Sep2} (or similar), step 4, 8 uses derivations similar to that in Eq.~\ref{I12}. In particular, in step 4 adding the projection $P$ does not affect the relation as the two tensors before adding the projection are related by a unitary on the left and down legs.

Therefore, we get
\begin{equation}
\text{rank}\left(\raisebox{-.4\height}{\includegraphics[height=1.0cm]{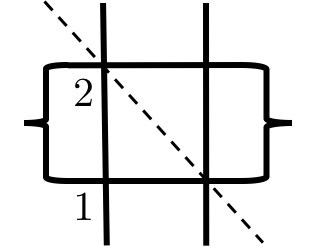}} \right) = \text{rank}\left(\raisebox{-.4\height}{\includegraphics[height=0.7cm]{r122}}\right)\text{rank}\left(\raisebox{-.4\height}{\includegraphics[height=0.7cm]{r123}}\right)
\label{r12injl}
\end{equation}
Similarly, we have
\begin{equation}
\text{rank}\left(\raisebox{-.4\height}{\includegraphics[height=1.0cm]{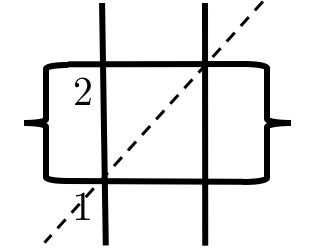}}\right) = \text{rank}\left(\raisebox{-.4\height}{\includegraphics[height=0.7cm]{r125}}\right)\text{rank}\left(\raisebox{-.4\height}{\includegraphics[height=0.7cm]{r126}}\right)
\label{r12injr}
\end{equation}
Dividing these two equations we get
\begin{equation}
I_{\text{RR}}(\tilde{\mathcal{M}}_{12}) = I_{\text{RR}}(\tilde{M}_1)I_{\text{RR}}(\tilde{M}_2)=I_{\text{RR}}(\tilde{M}_{12}).
\end{equation}

Therefore, the Rank-Ratio index of the stack MPUO $M_{12}$ remains the same whether we reduce it to the injective form or not and we always have
\begin{equation}
I_{\text{RR}}(\tilde{M}_{12}) = I_{\text{RR}}(\tilde{M}_1)I_{\text{RR}}(\tilde{M}_2).
\end{equation}
This property of the Rank-Ratio index is the same as that of the GNVW index which multiply when we combine two locality preserving unitaries (Eq.~\ref{III}). As in each of the injective layers (either representing local unitary or translation) the Rank-Ratio index is equal to the square of the GNVW index, when we stack the layers, the Rank-Ratio index is still equal to the square of the GNVW index. Therefore, for injective or stack of injective MPUO representations, we always have
\begin{equation}
I_{\text{RR}}(\tilde{M}) = (I_{\text{GNVW}}(O))^2
\end{equation}

5. Finally, we need to show that our definition of Rank-Ratio index is stable. That is, it does not change if we keep blocking the tensor $M$ once it has reached the fixed point form. This is true for both injective and stack of injective tensors.

Suppose that $M$ is at the fixed point form satisfying the separation, isometry, pulling through conditions in Eq.~\ref{separation}, ~\ref{isometry}, ~\ref{pulling_through}. Then we have
\begin{equation}
\begin{array}{ll}
& \lambda\left(\raisebox{-.4\height}{\includegraphics[height=1.2cm]{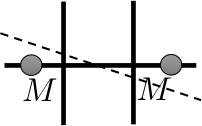}} \right)  = \lambda^{1/2}\left(\raisebox{-.4\height}{\includegraphics[height=1.6cm]{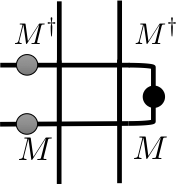}} \right)  \\ 
= & \lambda^{1/2}\left(\raisebox{-.4\height}{\includegraphics[height=1.6cm]{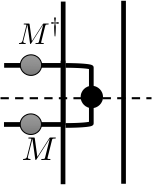}} \right)
= \lambda\left(\raisebox{-.4\height}{\includegraphics[height=1.2cm]{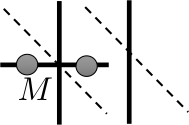}} \right)  
\end{array}
\label{IMM}
\end{equation}
Therefore, we have
\begin{equation}
\text{rank}\left(\raisebox{-.4\height}{\includegraphics[height=1.2cm]{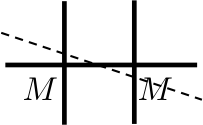}}\right) = \text{rank}\left(\raisebox{-.4\height}{\includegraphics[height=1.0cm]{lSVDM}}\right) \times d
\end{equation}
Similarly we have
\begin{equation}
\text{rank}\left(\raisebox{-.4\height}{\includegraphics[height=1.2cm]{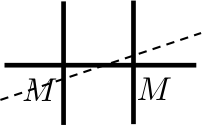}}\right) = \text{rank}\left(\raisebox{-.4\height}{\includegraphics[height=1.0cm]{rSVDM}}\right) \times d
\end{equation}
Dividing these two equations we find that the Rank-Ratio index does not change if we block tensors at the fixed point. Note that when $M$ is a stack of injective tensors, the grey dots in the previous equations actually correspond to several grey dots, one on each injective virtual leg.

With these steps, we complete the proof of Theorem~\ref{thm:RR_GNVW}. Note that, as our proof relies on Lemma~\ref{index_lemma} which is about the spectrum of the SVD decomposition, so in principle we can define our index as the ratio of the exponential of the entropy of the left and right SVD decompositions. The only tricky part is that we need to add the grey dots, the square root of the left and right eigenvectors of the transfer matrices, to the virtual legs for the index to work. This is doable but procedural-wise complicated. Therefore, we choose to define the index using the rank, instead of the entropy, of the SVD decomposition.
\end{proof}

\section{\label{sec:numerics} Numerical calculation of index for random MPUO} 
In this section we are going to calculate the rank-ratio index of some examples of random MPUO. The examples of random MPUO considered are drawn in Fig.~\ref{fig:exm}, and the corresponding numerical results are given in Tables ~\ref{tab:exm1},~\ref{tab:exm2} and \ref{tab:exm3} respectively. \par 
To generate random $k$-body unitaries we use the QR-decomposition of random matrices. The algorithm is as follows: 
\begin{enumerate}
    \item Generate $d^k $ dimensional random matrix $M_{d^k\times d^k}$. $d$ is the dimension of the physical Hilbert space at each site. 
    \item Perform a $QR$-decomposition: $M = QR$. $Q$ is a $d^k$ dimensional unitary while $R$ is an upper triangular matrix.  \item The $Q$ and $R$ are not unique since for any $d^k$ dimensional unitary diagonal matrix $\Lambda$, $QR = (Q\Lambda)( \Lambda^{-1}R)$. To fix this, we demand that $R$ has positive diagonal entries. This fixes $\Lambda$ to be identity. If $R=\sum_{ij}r_{ij}|i\rangle\langle j|$, create a diagonal matrix $\Lambda' = \sum_i \frac{r_{ii}}{|r_{ii}|} |i\rangle \langle i|$, and  $Q'=Q\Lambda'$. Now for every random matrix $M$, $Q'$ is a unique $d^k$ dimensional unitary. 
\end{enumerate}

From these examples, we can see that 
\begin{itemize}
\item The Rank-Ratio index fluctuates for small block sizes but saturates to a fixed value for large enough block sizes;
\item The saturated value is equal to the square of the GNVW index and only depends on the equivalence class of the MPUO which is invariant under stacking with any finite depth local unitary operation.
\end{itemize}

\begin{figure}
    \centering
    \includegraphics[width=\columnwidth]{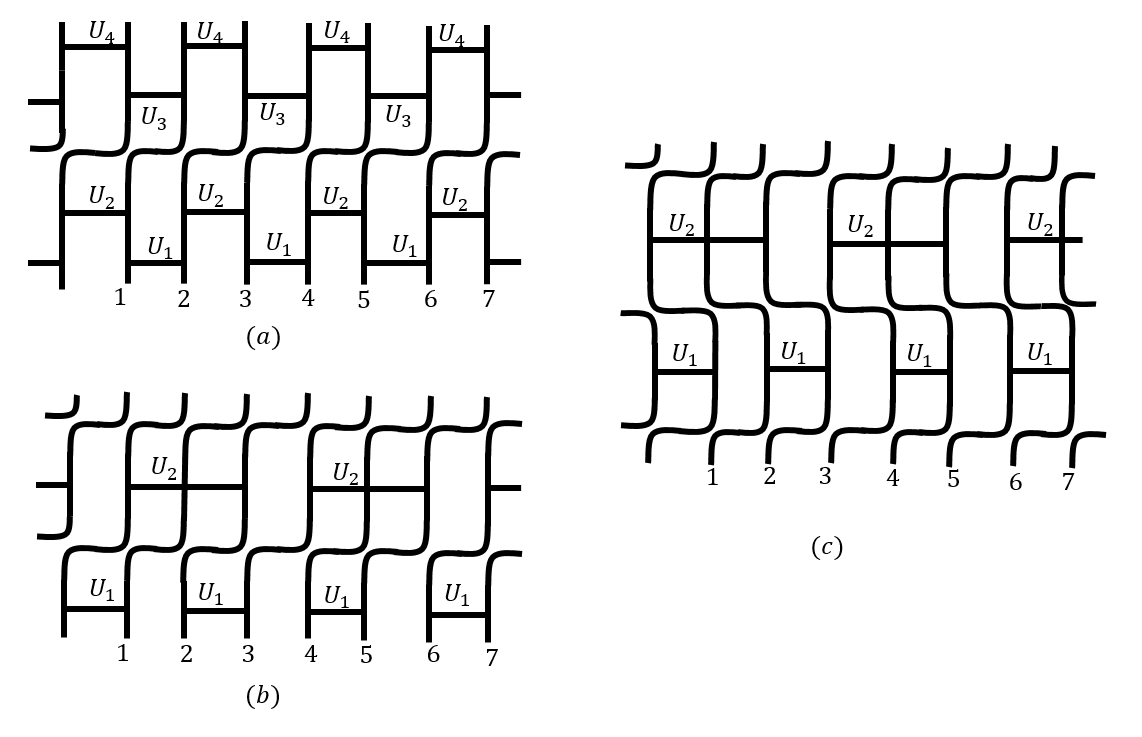}
    \caption{Some examples of random MPUOs. Local physical Hilbert space has dimension $d=2$ in all cases. (a) We combine a single right-translation operator with random finite depth local unitary operators. $U_1,U_2,U_3,U_4$ are all random 2-local unitaries,  (b) we combine layers of random local unitaries with layers of right-translation. First layer is made of 2-local random unitary $U_1$, second layer is right-translation, third layer is 3-local random unitary and fourth layer is again a right translation operator. (c) Finally as an example of the most general case we combine random local unitary operators with left and right translational operators. First layer is right-translation, second layer is random 2-local unitaries, third layer is left-translation, fourth layer is random 3-local untaries and final layer is right-translation again. Numerical calculation of RR indices of MPUOs in (a),(b) and (c) are given in tables ~\ref{tab:exm1},~\ref{tab:exm2} and \ref{tab:exm3} respectively. }
    \label{fig:exm}
\end{figure}
\begin{table}[h]
    \centering
    \begin{tabular}{c|c|c|c}
    \shortstack{Length of \\ blocked MPO}    & \shortstack{rank of \\ left SVD} &\shortstack{rank of \\ right SVD } & RR index  \\
     \hline 
    1 & 64 & 16 & 4\\
    2 & 8  & 8 & 1\\
    3 & 16 & 4 & 4\\
    4 & 32 & 8 & 4\\
    5 & 64 & 16 & 4\\
    6 & 128 & 32 & 4\\
    7 & 256 & 64 & 4\\
    \end{tabular}
    \caption{Numerical calculation of RR index of MPUO shown in Fig.~\ref{fig:exm}(a). We start with site labeled 1 and block sites one by one to the right. We see that after blocking 3 sites index stabilizes to value 4, which is expected since this MPUO is, by construction, equivalent (up to finite depth local untiaries) to a pure right-translation and hence has index  $I_{\text{RR}}(M_r)=2^2=4$. }
    \label{tab:exm1}
\end{table}

\begin{table}[h]
    \centering
    \begin{tabular}{c|c|c|c}
    \shortstack{Length of \\ blocked MPO}    & \shortstack{rank of \\ left SVD} &\shortstack{rank of \\ right SVD } & RR index  \\
     \hline 
    1 & 8   & 8 & 1\\
    2 & 16  & 4 & 4\\
    3 & 32  & 2 & 16\\
    4 & 64  & 4 & 16\\
    5 & 128 & 8 & 16\\
    6 & 256 & 16 & 16\\
    7 & 512 & 32 & 16\\
    \end{tabular}
    \caption{Numerical calculation of RR index of MPUO shown in Fig.~\ref{fig:exm}(b). We start with site labeled 1 and block sites one by one to the right. We see that after blocking 3 sites index stabilizes to value 16, which is expected since this MPUO is, by construction, equivalent (up to finite depth local untiaries) to the combination of two pure right-translation and hence has total index  $I_{\text{RR}}(M_r)^2=4^2=16.$ }
    \label{tab:exm2}
\end{table}

\begin{table}[h]
    \centering
    \begin{tabular}{c|c|c|c}
    \shortstack{Length of \\ blocked MPO}    & \shortstack{rank of \\ left SVD} &\shortstack{rank of \\ right SVD } & RR index  \\
     \hline 
    1    & 16 & 4 & 4\\
    2 & 32 & 8 & 4\\
    3 & 64 & 4 & 16\\
    4& 32 & 8 & 4\\
    5 & 64 & 16 & 4\\
    6 & 128 & 32 & 4\\
    7 & 512 & 128 & 4\\
    \end{tabular}
    \caption{Numerical calculation of RR index of MPUO shown in Fig.~\ref{fig:exm}(c). We start with site labeled 1 and block sites one by one to the right. We see that after blocking 3 sites index stabilizes to value 4, which is expected since this MPUO is, by construction, equivalent (up to finite depth local untiaries) to the combination of two pure right-translation and one left-translation, and hence has total index  $I_{\text{RR}}(M_r)I_{\text{RR}}(M_l)I_{\text{RR}}(M_r)=4.\frac{1}{4}.4=4$.}
    \label{tab:exm3}
\end{table}

\section{\label{sec:beyondGNVW} MPO with fractional index}

In the previous section, we have discussed how matrix product operators satisfying a simple unitary condition (Definition~\ref{def:MPUO} and Eq.~\ref{uni}) provides a necessary and sufficient representation of locality preserving unitaries classified by the GNVW index. On the other hand, if we relax the condition in Eq.~\ref{uni}, we can obtain matrix product operators, which are unitary in a more general sense, with index beyond the GNVW framework. In this section, we are going to give one example of such matrix product operators. We are going to show that this operator is unitary in systems of odd size and non-unitary in systems of even size. It does not preserve locality and can have a `fractional' index!

Consider the MPO $O_f$ represented with local tensor
\begin{equation}
M_f = \raisebox{-.4\height}{\includegraphics[height=1.5cm]{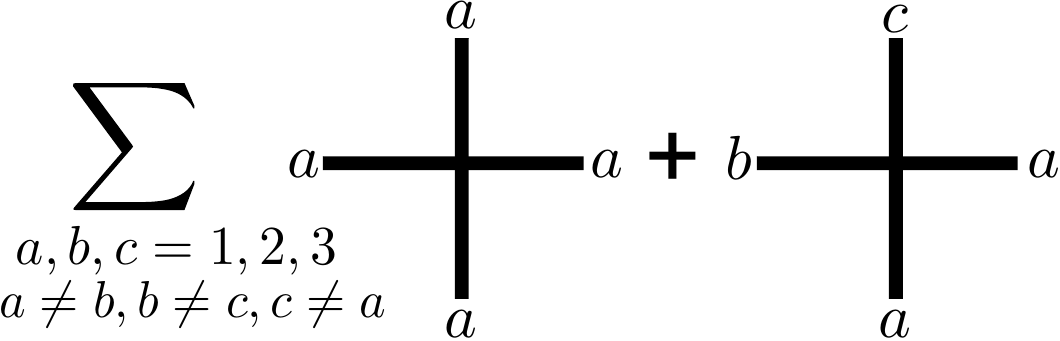}}
\end{equation}
This is a special MPO in that it represents a unitary operator when system size is odd and a non-unitary operator when system size is even. For example, when the system size is two, the operator maps both input states $|12\rangle$ and $|21\rangle$ to $|33\rangle$. Similar non-unitary mappings exist whenever the system size is even. This is different from all the other examples we discussed in this paper, which are unitary and satisfy Eq.~\ref{uni} for all system sizes. (And this operator does not satisfy Eq.~\ref{uni} even after blocking.) Therefore, it does not belong to the set of MPUO as defined in Definition~\ref{def:MPUO}.

To understand the property of this MPO, we can construct $T_f$ according to Eq.~\ref{Tij} and, from its general form, identify the operator $O_f^{\dagger}O_f$. The general form of $T_f$, which we calculate using the procedure in Ref.\onlinecite{Perez-Garcia2007}, contains two blocks. The first block is what we would expect if $O$ is a unitary for all system sizes \begin{equation}
\raisebox{-.4\height}{\includegraphics[height=2.0cm]{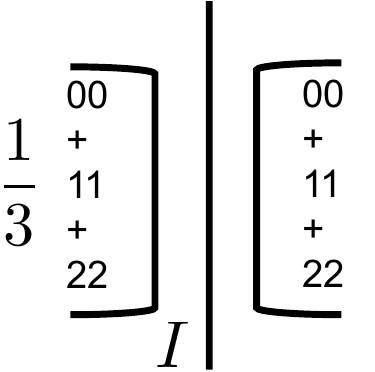}}
\end{equation}
Different from a usual unitary MPO, there is a second block, which represents the superposition of two translation symmetry breaking operators. The two operators each have period $2$ and they map into each other under a single step of translation. Therefore, this part of the MPO is zero when the system size is odd, leaving the MPO $O_f$ to be unitary. When the system size is even, the second block gives rise to a nontrivial operator, which breaks the unitarity of $O_f$.

When the system size is odd ($2n+1$), $O_f$ is a unitary operator, but it is a highly non-locality preserving. To see this, consider the operator $P_{n} = |1\rangle\langle 2| + |2\rangle\langle 3| + |3\rangle\langle 1|$ on the $n$th qutrit and the conjugation of $P_{n}$ by $O_f$. Apply $O^{\dagger}_fP_{n}O_f$ on an initial state $|11...1...11\rangle$, we find that the state is mapped to
\begin{equation}
\begin{array}{llll}
 & |11...1...11\rangle & \xrightarrow{O_f} & |11...1...11\rangle \\
 \xrightarrow{P_n} & |11...3...11\rangle & \xrightarrow{O^{\dagger}_f} & |32...a...32\rangle
 \end{array}
\end{equation}
where $a=3$ if $n$ is odd and $a=1$ if $n$ is even. As the final state $|32...a...32\rangle$ is globally different from the initial state $|11...1...11\rangle$, $O^{\dagger}_fP_{n}O_f$ has to be a nonlocal operator even tough $P_n$ is local. Therefore, $O_f$ is a non-locality-preserving unitary when system size is odd.

Interestingly, if we calculate the index of $M_f$ according to Eq.~\ref{index}, we find that
\begin{eqnarray}
I_{\text{RR}}(O_f) = \text{rank}\left(\raisebox{-.4\height}{\includegraphics[height=1cm]{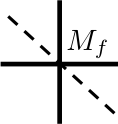}}\right)\bigg/\text{rank}\left(\raisebox{-.4\height}{\includegraphics[height=1cm]{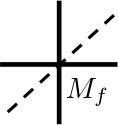}}\right)= 3
\end{eqnarray}
and this number stays invariant if we take blocks of $M_f$.If we were to convert it to the GNVW index, we would find it to be $\sqrt{3}$ which is not a rational number and hence not allowed as a GNVW index. This is of course expected because $O_f$ is not a locality preserving unitary and this example illustrates that it is possible to represent some non-locality-preserving unitaries with drastically different properties from the locality-preserving ones using the matrix product operator formalism.

\section{Conclusion}

In this paper we study the representation of one dimensional locality preserving unitaries using the matrix product operator (MPO) formalism. We show that matrix product operators which are unitary (for all system sizes) are guaranteed to preserve locality and all locality preserving unitaries can be represented in a matrix product way. Moreover, we show that the GNVW index\cite{Gross2012} classifying locality preserving unitaries in 1D can be extracted in a simple way as in Eq.~\ref{index} for injective or a stack of injective tensors. On the other hand, matrix product operators satisfying a more general unitarity condition -- unitary only for systems of certain sizes -- can have very different properties. In particular, we present one example of MPO which is unitary for odd size systems but not for even size systems and find that it is non-locality preserving and has a fractional index as compared to the locality preserving ones. 

Many interesting questions remain open regarding the matrix product representation of unitaries. First of all, Lemma~\ref{canonical_form} provides a complete characterization of MPOs which are unitary for any system size. However, this characterization is in terms of $T$ rather than $M$. In particular, if one wants to simulate a unitary evolution process using finite bond dimension MPO, it is not clear which parameter space one should choose from such that the MPO is guaranteed to be unitary. If such a parameter space can be identified, we can generate 1D unitaries without having to check the condition on the $T$ tensor. With the matrix product representation of states, we do not need to worry about this problem because any tensor generates a legitimate quantum state. This is essential for variational algorithms based on matrix product states. If we want to have similar simulation algorithms for unitary dynamics with matrix product operator, this problem needs to be addressed.

Secondly, adding symmetry requirement to the 1D unitary operators can result in more detailed classifications. This has been discussed in terms of (dynamical) interacting Floquet phases with symmetry where a classification in 1D has been proposed in Ref.\onlinecite{Else2016,Keyserlingk2016,Keyserlingk2016a,Potter2016,Roy2016}. Similar to the case of 1D gapped (nondynamical) phases, adding symmetry can result in symmetry-protected Floquet phases. It would be interesting to see how to distinguish different symmetry protected Floquet phases based on the MPO representation of their Floquet operator. 

Finally, the example we discussed in section~\ref{sec:beyondGNVW} shows that if we relax the definition of unitarity, MPO can represent non-locality-preserving unitaries with fractional index. What is the full power of MPO in representing 1D unitaries in this more general sense? For matrix product state, we know that with a translation invariant finite bond dimension representation, the state represented is either gapped or a superposition of several gapped states. Can we obtain a similar understanding of the MPO representation of 1D unitaries? This is a question we plan to study in the future.

\section*{Acknowledgment}
While working on this project, we became aware of similar work being carried out by J. I. Cirac, D. Perez-Garcia, N. Schuch, and F. Verstraete which has been reported in Ref.\onlinecite{Cirac2017}.

We would like to acknowledge helpful discussions with Lukasz Fidkowski which inspired this project. X.C. is supported by the National Science Foundation under award number DMR-1654340. We acknowledge funding provided by the Walter Burke Institute for Theoretical Physics and the Institute for Quantum Information and Matter, an NSF Physics Frontiers Center (NSF Grant PHY-1125565) with support of the Gordon and Betty Moore Foundation (GBMF-2644). M.B.S. acknowledges the support from Simons Qubit fellowship provided by Simons Foundation through It from Qubit collaboration.

\bibliography{main}

\end{document}